\documentclass[11pt,a4paper]{article}
\usepackage[utf8]{inputenc}
\usepackage[T1]{fontenc}
\usepackage{lmodern}
\usepackage{microtype}
\usepackage[margin=1.05in]{geometry}
\usepackage{amsmath,amssymb,amsthm,mathtools,bm}
\usepackage{enumitem}
\usepackage{booktabs,array}
\usepackage{xcolor}
\usepackage{dsfont}
\usepackage{tikz}
\usetikzlibrary{arrows.meta,calc,positioning}
\usepackage[colorlinks=true,linkcolor=black,citecolor=black,urlcolor=black]{hyperref}

\setlist{itemsep=2pt,topsep=3pt}
\hypersetup{pdftitle={Beyond IP = PSPACE and QIP = PSPACE: Interactive Proofs in Arbitrary Physical Theories},pdfauthor={Kishor Bharti}}

\theoremstyle{plain}
\newtheorem{theorem}{Theorem}[section]
\newtheorem{proposition}[theorem]{Proposition}
\newtheorem{lemma}[theorem]{Lemma}
\newtheorem{corollary}[theorem]{Corollary}
\newtheorem{fact}[theorem]{Fact}
\theoremstyle{definition}
\newtheorem{definition}[theorem]{Definition}
\newtheorem{example}[theorem]{Example}
\newtheorem{question}[theorem]{Open question}
\theoremstyle{remark}
\newtheorem{remark}[theorem]{Remark}

\newcommand{\cls}[1]{\mathsf{#1}}
\renewcommand{\P}{\cls{P}}
\newcommand{\NP}{\cls{NP}}
\newcommand{\PSPACE}{\cls{PSPACE}}
\newcommand{\EXP}{\cls{EXP}}
\newcommand{\NEXP}{\cls{NEXP}}
\newcommand{\RE}{\cls{RE}}
\newcommand{\BPP}{\cls{BPP}}
\newcommand{\PP}{\cls{PP}}
\newcommand{\MA}{\cls{MA}}
\newcommand{\AM}{\cls{AM}}
\newcommand{\QMA}{\cls{QMA}}
\newcommand{\IP}{\cls{IP}}
\newcommand{\QIP}{\cls{QIP}}
\newcommand{\MIP}{\cls{MIP}}
\newcommand{\MIPs}{\cls{MIP}^{*}}
\newcommand{\BQP}{\cls{BQP}}
\newcommand{\AWPP}{\cls{AWPP}}

\newcommand{\DSPACE}{\cls{DSPACE}}
\newcommand{\DTIME}{\cls{DTIME}}
\newcommand{\GIP}{\cls{GIP}}
\newcommand{\GIPcl}{\cls{GIP}_{\mathrm{cl}}}
\newcommand{\GIPqm}{\cls{GIP}_{\mathrm{qm}}}
\newcommand{\GMA}{\cls{GMA}}
\newcommand{\BGP}{\cls{BGP}}

\newcommand{\R}{\mathbb{R}}
\newcommand{\C}{\mathbb{C}}
\newcommand{\N}{\mathbb{N}}

\newcommand{\Vp}[1]{V^{+}_{#1}}
\newcommand{\Om}[1]{\Omega_{#1}}
\newcommand{\Eff}[1]{\mathcal{E}_{#1}}
\newcommand{\Chan}[2]{\mathrm{Chan}_{\mathcal{G}}(#1,#2)}
\newcommand{\GG}{\mathcal{G}}
\newcommand{\tmin}{\otimes_{\min}}
\newcommand{\tmax}{\otimes_{\max}}
\newcommand{\Herm}{\mathrm{Herm}}
\newcommand{\Pos}{\mathrm{Pos}}
\newcommand{\Sym}{\mathrm{Sym}}
\newcommand{\Tr}{\mathrm{Tr}}
\newcommand{\val}{\mathrm{val}}
\newcommand{\id}{\mathrm{id}}
\newcommand{\poly}{\mathrm{poly}}
\newcommand{\acc}{\mathrm{acc}}
\newcommand{\conv}{\mathrm{conv}}
\newcommand{\cone}{\mathrm{cone}}
\newcommand{\ket}[1]{|#1\rangle}
\newcommand{\bra}[1]{\langle #1|}
\newcommand{\ketbra}[1]{|#1\rangle\langle #1|}
\newcommand{\QC}{\mathsf{QC}}
\newcommand{\eps}{\varepsilon}
\newcommand{\one}{\mathds{1}}

\title{\textbf{Beyond $\IP = \PSPACE$ and $\QIP = \PSPACE$: Interactive Proofs in Arbitrary Physical Theories}}
\author{Kishor Bharti$^{1,2,3\,*}$\\[6pt]
\small $^{1}$QuICS, NIST/University of Maryland, College Park, MD 20742, USA\\
\small $^{2}$UMIACS, University of Maryland, College Park, Maryland 20742, USA\\
\small $^{3}$IHPC, A*STAR, Singapore\\[4pt]
\small $^{*}$Currently at IonQ, Inc.}
\date{September 2026}

\begin{document}
\maketitle
\begin{abstract}
The equalities $\IP = \PSPACE$ and $\QIP = \PSPACE$, the latter achievable with three messages, raise a basic question: how much of an interactive proof’s power comes from the underlying physical theory? We study interactive proofs in general probabilistic theories, which include classical and quantum theory. The answer depends on what the prover and verifier exchange and how the theory specifies efficient operations. When they exchange only classical messages, protocols in every theory satisfying our standard assumptions decide exactly $\PSPACE$. For protocols with a quantum verifier and quantum messages, allowing a prover to use any theory containing quantum theory does not increase the maximum acceptance probability. Thus, the three-message $\PSPACE$ result remains valid against such provers. When messages may be arbitrary systems, the interactive-proof class can strictly exceed $\PSPACE$.
\end{abstract}

\section{Introduction}\label{sec:intro}

An interactive proof system lets a verifier with bounded computational power check a claim by exchanging messages with a prover whose power is unbounded \cite{GMR89,Bab85}. Completeness requires that an honest prover convinces the verifier on yes-instances; soundness requires that no prover, honest or not, convinces it on no-instances except with small probability. Two theorems describe the power of this model. Classical interactive proofs decide exactly the problems solvable in polynomial space,
\[
\IP=\PSPACE\quad\text{\cite{LFKN92,Sha92,She92}},
\]
and quantum interactive proofs, in which the verifier is a polynomial-time quantum circuit and the messages are quantum registers \cite{Wat03,KW00}, decide the same class,
\[
\QIP=\PSPACE\quad\text{\cite{JJUW11}}.
\]
Quantum interaction does not enlarge the class, but it shortens the interaction: three quantum messages suffice, $\QIP(3)=\QIP$ \cite{KW00}, whereas any constant number $k\ge2$ of classical messages gives $\AM$ \cite{GS86,BM88}, which equals $\PSPACE$ only if the polynomial hierarchy collapses. Standard accounts are given in Refs.~\cite{AB09,Gol08,VW16,Wat09}.

These results raise a basic question. The verifier and the prover are physical agents, and the operations they perform and the systems they exchange are governed by a physical theory: classical probability theory in $\IP$, quantum theory in $\QIP$. How much of the power of an interactive proof comes from this theory? Could a prover with stronger-than-quantum correlations convince a quantum verifier of a false statement? Could a verifier equipped with such a theory check more than $\PSPACE$? For computation without interaction, questions of this kind have concrete answers: nonlinear modifications of quantum mechanics would solve $\NP$-complete problems in polynomial time \cite{AL98}, postselection raises $\BQP$ to $\PP$ \cite{Aar05}, and quantum circuits with arbitrary real amplitudes decide undecidable problems \cite{ADH97}; see Ref.~\cite{Aar04} for a discussion of this theme.

We study interactive proofs in general probabilistic theories. This framework, which grew out of the operational approach to quantum foundations and is developed in Refs.~\cite{Har01,Bar07,CDP10,CDP11,JH14,Mul21,Pla23,CS16}, describes a theory by its systems, states, measurements, transformations and composite systems. Classical and quantum theory are instances, and so are theories with stronger-than-quantum correlations such as Barrett's boxworld \cite{Bar07}, which contains the Popescu--Rohrlich box \cite{PR94}. Lee and Barrett introduced polynomial-time uniform circuits in this framework and proved that every theory satisfying a few assumptions has a bounded-error class contained in $\AWPP\subseteq\PP$ \cite{LB15}; later work studied advice and single-message proofs \cite{LH16}, query complexity \cite{LB16,LS16,BLS18}, reversible circuit models \cite{KM19} and {the complexity classes attainable in different theories} \cite{BdBHL19}. Here we allow polynomially many rounds of interaction between a verifier that runs a polynomial-time uniform family of circuits of the theory and a prover that may use any system and any transformation of the theory, with no computational bound.

\subsection{Results}\label{sec:intro-results}

We write $\GIP(\GG)$ for the class of promise problems that have an interactive proof in the theory $\GG$ with completeness $2/3$ and soundness $1/3$, and $\GIP(\GG;k)$ for the subclass with at most $k$ messages. Every theory we consider satisfies four standing assumptions, stated in Section~\ref{sec:standing}: local tomography, causality, containment of classical theory, and computability of the verifier's gates. Under these assumptions the answer to the question above depends on two things: which systems the verifier and the prover exchange, and how the theory specifies its efficient operations.

\paragraph{Classical messages.}
Let $\GIPcl(\GG)$ be the subclass of $\GIP(\GG)$ in which every message is a classical register. The verifier may still prepare, transform and measure arbitrary systems of $\GG$ internally, and the prover may use arbitrary systems and transformations of $\GG$.

\begin{theorem}[Theorem~\ref{thm:cl-universal}]\label{thm:intro-A}
For every theory $\GG$ satisfying the standing assumptions, $\GIPcl(\GG)=\PSPACE$.
\end{theorem}

The theory enters only through the verifier's own circuit; the prover's nonclassical resources are of no use to it. Conditioned on the classical transcript, the verifier's and the prover's private systems remain in a product state, so the prover's behaviour is described by classical conditional distributions of answers given questions, and a deterministic prover is optimal. The value of the protocol is then an alternating sum-and-maximum expression of polynomial depth whose leaves are acceptance probabilities of the verifier's circuit with the answers fixed, and these leaves can be approximated in polynomial space by the path-sum method of Lee and Barrett \cite{LB15}. This extends the argument of Feldman and of Papadimitriou for $\IP\subseteq\PSPACE$ \cite{Fel86,Pap85} to arbitrary verifier circuits.

\paragraph{Quantum verifiers.}
Let $\GG$ contain quantum theory, in the sense that its qubit registers are ordinary quantum registers on which all quantum channels are allowed, and let $\GIPqm(\GG)$ be the subclass of $\GIP(\GG)$ in which the verifier is a polynomial-time quantum circuit and all messages are qubit registers. The prover may use arbitrary systems and channels of $\GG$, including systems whose joint states with the verifier's qubits lie outside the quantum state space.

\begin{theorem}[Theorem~\ref{thm:qm-universal}]\label{thm:intro-B}
Let $\GG$ contain quantum theory. For every quantum verifier $V$ and every input $x$, the maximum acceptance probability of $V_x$ over provers of $\GG$ equals its maximum over quantum provers. Consequently $\GIPqm(\GG;k)=\QIP(k)$ for every $k$; in particular $\GIPqm(\GG)=\GIPqm(\GG;3)=\PSPACE$.
\end{theorem}

A prover's strategy induces a linear map from the question registers to the answer registers. Positivity on all quantum extensions makes this map completely positive by Choi's theorem \cite{Cho75}, and causality gives it the normalisation constraints of a quantum strategy in the sense of Gutoski and Watrous \cite{GW07}, equivalently of a quantum comb in the sense of Chiribella, D'Ariano and Perinotti \cite{CDP08,CDP09}. The realisation theorem for quantum strategies then supplies a quantum prover with the same acceptance probability. Thus the three-message characterisation $\QIP(3)=\PSPACE$ remains valid against provers that are not confined to quantum theory.

\paragraph{Arbitrary messages.}
When the messages may be arbitrary systems of the theory, the class need not be $\PSPACE$, even in theories that contain both quantum theory and classical theory.

\begin{theorem}[Theorems~\ref{thm:triangle}, \ref{thm:Qtriangle} and~\ref{thm:computable}]\label{thm:intro-C}
For every real number $v\in(0,1)$ there is a theory $\GG^{\mathrm Q}_v$ satisfying the standing assumptions, containing quantum theory and classical theory as subtheories, and isomorphic to quantum theory with classical systems, such that $\GIPcl(\GG^{\mathrm Q}_v)=\GIPqm(\GG^{\mathrm Q}_v)=\PSPACE$ and:
\begin{enumerate}[label=(\alph*)]
\item if $v$ is not computable, then a protocol with a single message from the prover decides an undecidable promise problem, so $\GIP(\GG^{\mathrm Q}_v)\supsetneq\PSPACE$;
\item there is a computable $v$ for which $\GIP(\GG^{\mathrm Q}_v)\supsetneq\PSPACE$ unconditionally; the separating problem is decidable in quasi-polynomial space.
\end{enumerate}
\end{theorem}

The theory $\GG^{\mathrm Q}_v$ is quantum theory with one additional classical three-level system $W_v$, whose state space is presented as the triangle with vertices $(0,0,1)$, $(1,0,1)$ and $(1,v,1)$. The verifier's gates have integer and dyadic coordinates in this presentation, so the standing assumptions hold whatever $v$ is; but the maximum of the verifier's measurement $f_y(x,y,z)=y$ over the states of $W_v$ equals $v$. A prover that sends many copies of the vertex $(1,v,1)$ lets the verifier estimate $v$ to inverse-polynomial precision, and comparing $v$ with the input decides a promise problem whose difficulty is that of computing $v$. For noncomputable $v$ this problem is undecidable; for a computable $v$ built from the space hierarchy theorem \cite{HLS65} it lies outside $\PSPACE$. The same construction with classical theory in place of quantum theory gives a theory $\GG_v$, isomorphic to classical theory, with the same properties.

\paragraph{Dependence on the specification of efficient operations.}
The theories $\GG_v$ and $\GG^{\mathrm Q}_v$ are isomorphic, as theories, to classical theory and to quantum theory with classical systems, whose interactive-proof classes are $\IP=\PSPACE$ and $\QIP=\PSPACE$ (Lemma~\ref{lem:QC}). {Under these isomorphisms, the designated measurement becomes a measurement using a coin of bias $v$, which need not be an efficient gate in the standard presentation.} If $v$ can be approximated efficiently, the two presentations are computationally equivalent and $\GIP(\GG_v)=\GIP(\GG^{\mathrm Q}_v)=\PSPACE$ (Proposition~\ref{prop:presentation}). Thus $\GIP(\GG)$ is not an invariant of the operational content of $\GG$ alone: it depends on how the theory specifies its efficient operations, and an upper bound on $\GIP(\GG)$ must use a hypothesis about this specification. Table~\ref{tab:summary} in Section~\ref{sec:discussion} summarises the results.

\subsection{Organisation}
Section~\ref{sec:prelim} recalls the complexity classes and the convex geometry we use. Section~\ref{sec:gpt} introduces general probabilistic theories, with the classical, quantum, boxworld and triangular theories as running examples, and states the standing assumptions. Section~\ref{sec:circuits} recalls circuits and the path-sum representation of their outcome probabilities. Section~\ref{sec:model} defines interactive proofs in a theory and proves error reduction. Sections~\ref{sec:cl}, \ref{sec:qm} and~\ref{sec:beyond} prove Theorems~\ref{thm:intro-A}, \ref{thm:intro-B} and~\ref{thm:intro-C}, and Section~\ref{sec:discussion} discusses the results and open questions. Appendix~\ref{app:hierarchy} contains the unary space hierarchy theorem used in Section~\ref{sec:beyond}.

\section{Preliminaries}\label{sec:prelim}

\paragraph{Notation.}
Complexity classes are written in sans-serif. All classes are classes of promise problems. A promise problem is a pair $L=(L_{\mathrm{yes}},L_{\mathrm{no}})$ of disjoint subsets of $\{0,1\}^*$; an algorithm decides it if it accepts every string of $L_{\mathrm{yes}}$ and rejects every string of $L_{\mathrm{no}}$, and a promise problem is decidable if some algorithm that halts on every input decides it. Logarithms are natural unless written $\log_2$. For a finite set $S$, $\R^S$ is the space of real functions on $S$ and $\delta_s\in\R^S$ is the indicator of $s$. Vectors of a real vector space $V$ are written $\omega,\sigma,\tau,\dots$; elements of the dual space $V^*$ are written $f,g,u,\dots$; the pairing is $f(\omega)$. A polynomial $p(n)$ is a fixed polynomial with nonnegative integer coefficients evaluated at the input length $n$, and $\poly(n)$ denotes some such polynomial.

\subsection{Complexity classes and interactive proofs}\label{sec:classes}

Standard references are Refs.~\cite{AB09,Gol08,HO02}; for the quantum classes see Refs.~\cite{Wat09,VW16,Wat18}. $\P$, $\NP$, $\PSPACE$ and $\EXP$ are deterministic polynomial time, nondeterministic polynomial time, polynomial space and deterministic time $2^{\poly(n)}$; $\DSPACE(s)$ and $\DTIME(t)$ are the deterministic space and time classes; $\BPP$ and $\BQP$ are bounded-error probabilistic and quantum polynomial time; $\PP$ is unbounded-error probabilistic polynomial time. We use $\PSPACE\subseteq\EXP$ and the fact that every problem in $\EXP$, and hence in $\PSPACE$, is decidable. The class $\AWPP$ \cite{FFK94,FR99} satisfies $\BQP\subseteq\AWPP\subseteq\PP\subseteq\PSPACE$; it appears only in the discussion of related work.

\begin{definition}[Interactive proofs]\label{def:IP}
A classical interactive proof system for a promise problem $L$ consists of a probabilistic polynomial-time verifier $V$ that exchanges polynomially many messages of polynomial length with a prover $P$ (an arbitrary, possibly randomised, function from the conversation so far to the next message) and then accepts or rejects, such that for every input $x$: (completeness) if $x\in L_{\mathrm{yes}}$, there is a prover with $\Pr[(V,P)(x)\text{ accepts}]\ge2/3$; (soundness) if $x\in L_{\mathrm{no}}$, then for every prover $\Pr[(V,P)(x)\text{ accepts}]\le1/3$. $\IP$ is the class of promise problems having such a system, $\IP(k)$ the subclass in which at most $k$ messages are exchanged, and $\MA=\IP(1)$ with the message sent by the prover.
\end{definition}

The thresholds $2/3$ and $1/3$ can be replaced by $1-2^{-p(n)}$ and $2^{-p(n)}$ for any polynomial {$p$ with $p(n)\ge2$}, and by any constants {$0<s<c<1$}, without changing $\IP${. To amplify a protocol with completeness $c$ and soundness $s$, run it sequentially polynomially many times and accept if the fraction of accepting runs exceeds $(c+s)/2$} (see, e.g., Ref.~\cite[Section~9.1]{Gol08}). Sequential repetition increases the number of messages; for the classes with a bounded number of messages we only use the following facts.

\begin{fact}[Lund, Fortnow, Karloff and Nisan; Shamir; Shen \cite{LFKN92,Sha92,She92}]\label{fact:IP}
$\IP=\PSPACE$. Every promise problem in $\PSPACE$ has an interactive proof system in which the verifier's messages are uniformly random strings (a public-coin system), with perfect completeness and soundness $1/3$ against arbitrary randomised provers.
\end{fact}

\begin{fact}[Babai and Moran; Goldwasser and Sipser \cite{BM88,GS86}]\label{fact:AM}
For every constant $k\ge2$, $\IP(k)=\AM$, the class of two-message public-coin interactive proofs, and $\AM\subseteq\Pi^p_2\subseteq\PSPACE$.
\end{fact}

\begin{definition}[Quantum interactive proofs]\label{def:QIP}
$\QIP$ is defined as in Definition~\ref{def:IP} except that the verifier is a polynomial-time uniform family of quantum circuits with intermediate measurements and classical control \cite{AKN98}, the prover may apply arbitrary quantum channels on a finite-dimensional private memory together with the message registers, and the messages are registers of polynomially many qubits. $\QIP(k)$ is the subclass with at most $k$ messages, and $\QMA=\QIP(1)$ with the message sent by the prover \cite{Kni96,KSV02,Wat18,VW16}.
\end{definition}

\begin{fact}[Kitaev and Watrous; Jain, Ji, Upadhyay and Watrous \cite{KW00,JJUW11}]\label{fact:QIP}
$\QIP=\QIP(3)=\PSPACE$.
\end{fact}

\subsection{Convex geometry}\label{sec:convex}

Let $V$ be a finite-dimensional real vector space with dual $V^*$. A cone $C\subseteq V$ is a set closed under multiplication by nonnegative scalars; it is convex if $C+C\subseteq C$, pointed if $C\cap(-C)=\{0\}$, and generating if $C-C=V$. A closed convex pointed cone induces the partial order $\omega\le_C\sigma$ iff $\sigma-\omega\in C$. The dual cone is $C^*=\{f\in V^*:f(\omega)\ge0\ \forall\omega\in C\}$. We use the following standard facts \cite{Roc70}; see also Ref.~\cite{Pla23} for their use in the present context.

\begin{fact}\label{fact:convex}
Let $C\subseteq V$ be a closed convex cone.
\begin{enumerate}[label=(\roman*)]
\item $(C^*)^*=C$ (bipolar theorem). Consequently, if $\omega\notin C$, there is $f\in C^*$ with $f(\omega)<0$.
\item If $C$ is pointed and generating, so is $C^*$.
\item If $u\in C^*$ is strictly positive on $C\setminus\{0\}$, then $\Omega=\{\omega\in C:u(\omega)=1\}$ is a compact convex set, every element of $C\setminus\{0\}$ is uniquely of the form $\lambda\omega$ with $\lambda>0$ and $\omega\in\Omega$, and $C$ is generating if and only if $\Omega$ spans $V$.
\item A compact convex set is the convex hull of its extreme points, and a linear functional attains its maximum over it at an extreme point.
\end{enumerate}
\end{fact}

For spaces $V_A,V_B$ with closed convex pointed generating cones $C_A,C_B$, the minimal and maximal tensor products are
\begin{align}
C_A\tmin C_B&=\cone\{\alpha\otimes\beta:\alpha\in C_A,\ \beta\in C_B\},\label{eq:minmax}\\
C_A\tmax C_B&=\{\xi\in V_A\otimes V_B:(f\otimes g)(\xi)\ge0\ \forall f\in C_A^*,\ g\in C_B^*\}.\nonumber
\end{align}
Both are closed convex cones (the first is the cone over the convex hull of the compact set of products of base elements), $C_A\tmin C_B\subseteq C_A\tmax C_B$, and $(C_A\tmin C_B)^*=C_A^*\tmax C_B^*$ directly from the definitions; hence $(C_A\tmax C_B)^*=C_A^*\tmin C_B^*$ by the bipolar theorem. The same definitions apply to more than two factors.

\section{General probabilistic theories}\label{sec:gpt}

We follow the conventions of Pl\'avala \cite{Pla23} and of Lee and Barrett \cite{LB15}, restricting attention to finite-dimensional systems. The reader familiar with quantum information may keep the following dictionary in mind: states generalise density matrices, effects generalise POVM elements, the unit effect generalises the trace, channels generalise completely positive trace-preserving maps, and ``positive on all extensions'' generalises complete positivity.

\subsection{Systems, states and measurements}

\begin{definition}[System]\label{def:system}
A system $A$ is a triple $(V_A,\Vp{A},u_A)$, where $V_A$ is a finite-dimensional real vector space, $\Vp{A}\subseteq V_A$ is a closed, convex, pointed and generating cone, and $u_A\in V_A^*$ is strictly positive on $\Vp{A}\setminus\{0\}$. The states of $A$ are the elements of $\Om{A}=\{\omega\in\Vp{A}:u_A(\omega)=1\}$, a compact convex set by Fact~\ref{fact:convex}(iii). The effects of $A$ are the elements of $\Eff{A}=\{f\in V_A^*:0\le f(\omega)\le1\ \forall\omega\in\Om A\}=\{f\in V_A^*:0\le f\le u_A\}$, where the order is that of the dual cone $(\Vp A)^*$; $u_A$ is the unit effect. A measurement with outcomes $\{1,\dots,k\}$ is a family $(f_1,\dots,f_k)$ of effects with $\sum_if_i=u_A$; on the state $\omega$ it produces outcome $i$ with probability $f_i(\omega)$. A state is pure if it is an extreme point of $\Om A$. The trivial system $I$ has $V_I=\R$, $\Vp I=\R_{\ge0}$ and $u_I=\id$.
\end{definition}

The definition encodes two operational ideas. Mixing preparations with probabilities $\lambda$ and $1-\lambda$ gives the state $\lambda\omega+(1-\lambda)\sigma$, so state spaces are convex; outcome probabilities respect mixtures, so effects are affine on $\Om A$, hence linear on $V_A$.

\begin{example}[Classical systems]\label{ex:classical}
The classical $n$-level system $\Delta_n$ has $V=\R^n$, $V^+=\R^n_{\ge0}$ and $u(\omega)=\sum_i\omega_i$. Its states are the probability distributions on $\{1,\dots,n\}$, its pure states the point masses $\delta_1,\dots,\delta_n$, its effects the vectors in $[0,1]^n$, and its $k$-outcome measurements the column-stochastic $k\times n$ matrices. $\Delta_2$ is a classical bit: its states are the vectors $(p,1-p)$, and the effect $(1,0)$ returns the probability $p$.
\end{example}

\begin{example}[Quantum systems]\label{ex:quantum}
The quantum $d$-level system $Q_d$ has $V=\Herm_d(\C)$, the real space of $d\times d$ Hermitian matrices, $V^+=\Pos_d$, the cone of positive semidefinite matrices, and $u=\Tr$. Its states are the density matrices, its pure states the rank-one projections, its effects the operators $0\le E\le\one$, and its measurements the POVMs; the probability of the outcome $E$ on the state $\rho$ is $\Tr(E\rho)$. $Q:=Q_2$ is the qubit. We use the same notation $f(\omega)$ for the pairing in both examples, so that for a qubit $f(\rho)=\Tr(E\rho)$ with $f\leftrightarrow E$.
\end{example}

\begin{example}[The gbit]\label{ex:gbit}
Let $V=\R^3$ with coordinates $(x,y,z)$, $u(x,y,z)=z$, and let $V^+$ be the cone over the square $\Omega=\{(x,y,1):|x|\le1,\ |y|\le1\}$. This system, introduced by Barrett \cite{Bar07}, has the four pure states $(\pm1,\pm1,1)$ and two distinguished two-outcome measurements, $X=(\tfrac12(u+x),\tfrac12(u-x))$ and $Z=(\tfrac12(u+y),\tfrac12(u-y))$, where $x$ and $y$ denote the coordinate functionals. On a pure state both measurements have deterministic outcomes, and every pair of outcomes occurs for exactly one pure state. The dual cone is $\{(g_x,g_y,g_z):g_z\ge|g_x|+|g_y|\}$; its extreme rays are generated by the four outcome effects of $X$ and $Z$. Consequently $X$ and $Z$ have no joint measurement: in a four-outcome measurement $(h_{ab})_{a,b=\pm}$ with $X$ and $Z$ as marginals, each $h_{ab}$ lies below $\tfrac12(u+ax)$ and below $\tfrac12(u+by)$ in the order of the dual cone, hence is a nonnegative multiple of each of these two linearly independent extreme elements, hence zero; but then $h_{++}+h_{+-}=\tfrac12(u+x)\ne0$ is contradicted. The gbit is the elementary system of boxworld (Example~\ref{ex:boxworld}).
\end{example}

\begin{example}[The triangular system $W_v$]\label{ex:triangle}
Fix a real number $v\in(0,1)$. The triangular system $W_v$ has $V=\R^3$, $u(x,y,z)=z$, and $V^+$ the cone generated by
\[
e_0=(0,0,1),\qquad e_1=(1,0,1),\qquad e_2=(1,v,1).
\]
Its state space is the triangle $\conv\{e_0,e_1,e_2\}$ (Figure~\ref{fig:triangle}). A triangle is a simplex, so as a single system $W_v$ is order-isomorphic to the classical three-level system: the linear map $L_v:\R^3\to\R^3$ with $L_v(e_i)=\delta_{i+1}$ maps $\Vp{W_v}$ onto $\R^3_{\ge0}$ and satisfies $u_{\Delta_3}\circ L_v=u$. What distinguishes $W_v$ from $\Delta_3$ is its coordinate representation. The functional $f_y(x,y,z)=y$ takes the values $0,0,v$ on $e_0,e_1,e_2$, so $0\le f_y\le u$ on $V^+$, and $(f_y,u-f_y)$ is a two-outcome measurement whose coordinate vectors are $(0,1,0)$ and $(0,-1,1)$; its maximum $\max_{\Om{W_v}}f_y=v$ is attained at $e_2$. In the canonical coordinates of $\Delta_3$ the same effect reads $f_y\circ L_v^{-1}=(0,0,v)$. The three-outcome measurement $(u-x,\ x-y/v,\ y/v)$ reads the vertex label perfectly, since its $i$-th effect takes the value $1$ on $e_i$ and $0$ on the other vertices; two of its coordinate vectors involve $1/v$. Thus when $v$ is not computable, the measurement $(f_y,u-f_y)$ has computable coordinates although the isomorphism $L_v$ and the label-reading measurement do not.
\end{example}

\begin{figure}[t]
\centering
\begin{tikzpicture}[x=3.2cm,y=3.2cm]
  \draw[-{Stealth[length=2mm]}] (-0.15,0)--(1.35,0) node[right] {$x$};
  \draw[-{Stealth[length=2mm]}] (0,-0.15)--(0,0.95) node[above] {$y$};
  \fill[gray!15] (0,0)--(1,0)--(1,0.6)--cycle;
  \draw[thick] (0,0)--(1,0)--(1,0.6)--cycle;
  \fill (0,0) circle[radius=1.4pt] node[below left] {$e_0=(0,0,1)$};
  \fill (1,0) circle[radius=1.4pt] node[below right] {$e_1=(1,0,1)$};
  \fill (1,0.6) circle[radius=1.4pt] node[above] {$e_2=(1,v,1)$};
  \draw[densely dashed] (0,0.6)--(1,0.6);
  \node[left] at (0,0.6) {$v$};
  \node[align=left,anchor=west] at (1.7,0.3) {$f_y(x,y,1)=y$\\[2pt] $\max_{\Om{W_v}} f_y=v$, attained at $e_2$\\[2pt] $(f_y,u-f_y)=\big((0,1,0),(0,-1,1)\big)$};
\end{tikzpicture}
\caption{The state space of the triangular system $W_v$ in the plane $z=1$. The verifier's measurement $(f_y,u-f_y)$ has integer coordinates for every $v$, but its maximum over the states is the real number $v$.}
\label{fig:triangle}
\end{figure}
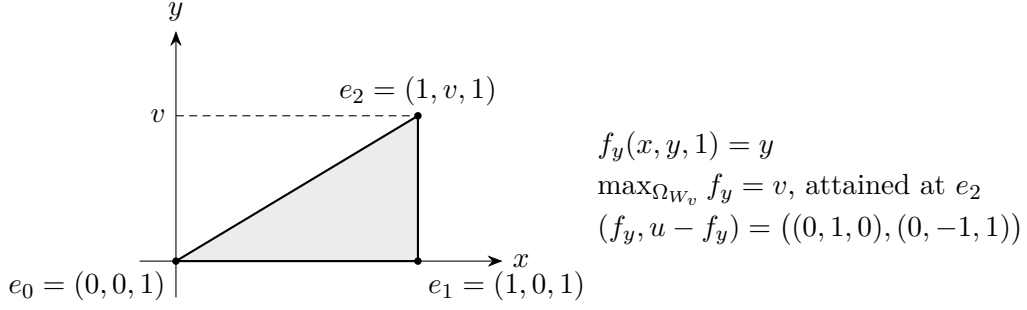

\subsection{Composite systems}

\begin{definition}[Composite]\label{def:composite}
A composite of systems $A$ and $B$ is a system $AB$ with
\[
V_{AB}=V_A\otimes V_B,\qquad u_{AB}=u_A\otimes u_B,\qquad \Vp A\tmin\Vp B\subseteq\Vp{AB}\subseteq\Vp A\tmax\Vp B,
\]
where $\Vp{AB}$ is a closed convex cone. For $\alpha\in\Om A$ and $\beta\in\Om B$ the product state $\alpha\otimes\beta$ is a state of $AB$, and for $f\in\Eff A$ and $g\in\Eff B$ the product effect $f\otimes g$ is an effect of $AB$.
\end{definition}

The equation $V_{AB}=V_A\otimes V_B$ is \emph{local tomography} \cite{Ara80,Har01,Bar07}: a state of $AB$ is determined by the probabilities of all product effects, because the product effects span $V_{AB}^*$. Within the two inclusions a theory is free to choose $\Vp{AB}$, and this choice determines which correlations, in particular which entangled states, the theory admits.

\begin{fact}[\cite{Bar07,Pla23}]\label{fact:composite-is-system}
With $V_{AB},u_{AB},\Vp{AB}$ as in Definition~\ref{def:composite}, $(V_{AB},\Vp{AB},u_{AB})$ is a system.
\end{fact}
\begin{proof}
$\Vp{AB}$ is generating because it contains all products $\alpha\otimes\beta$ of elements of the generating cones $\Vp A$ and $\Vp B$. It is pointed because $\xi,-\xi\in\Vp A\tmax\Vp B$ forces $(f\otimes g)(\xi)=0$ for all $f\in(\Vp A)^*$ and $g\in(\Vp B)^*$, and such products span $V_{AB}^*$ since the dual cones are generating (Fact~\ref{fact:convex}(ii)). For strict positivity of $u_{AB}$, let $\xi\in\Vp A\tmax\Vp B$ with $(u_A\otimes u_B)(\xi)=0$. For effects $f\in\Eff A$ and $g\in\Eff B$ the four products $f\otimes g$, $(u_A-f)\otimes g$, $f\otimes(u_B-g)$ and $(u_A-f)\otimes(u_B-g)$ are nonnegative on $\xi$ and sum to $(u_A\otimes u_B)(\xi)=0$, so $(f\otimes g)(\xi)=0$ for all effects $f,g$. Every element of $(\Vp A)^*$ is a nonnegative multiple of an effect (scale it by the reciprocal of its maximum over the compact set $\Om A$), so effects span $V_A^*$, and likewise for $B$; hence $\xi=0$.
\end{proof}

\begin{fact}[Marginals and conditional states \cite{Bar07,Pla23}]\label{fact:marginal}
Let $AB$ be a composite, $\omega\in\Vp{AB}$ and $g\in(\Vp B)^*$. Then $(\id_A\otimes g)(\omega)\in\Vp A$. For $\omega\in\Om{AB}$ and $g\in\Eff B$ we can write $(\id_A\otimes g)(\omega)=p\,\omega'$ with $p=(u_A\otimes g)(\omega)\in[0,1]$ and, if $p>0$, a unique $\omega'\in\Om A$: the state of $A$ conditioned on the outcome $g$. For $g=u_B$ this is the marginal state $\omega_A$.
\end{fact}
\begin{proof}
For $f\in(\Vp A)^*$, $f\big((\id_A\otimes g)(\omega)\big)=(f\otimes g)(\omega)\ge0$ because $\omega\in\Vp A\tmax\Vp B$; by the bipolar theorem $(\id_A\otimes g)(\omega)\in\Vp A$. The rest is Fact~\ref{fact:convex}(iii) together with $(u_A\otimes g)(\omega)\le(u_A\otimes u_B)(\omega)=1$.
\end{proof}

\begin{example}[Minimal and maximal tensor products]\label{ex:minmax}
If $\Vp{AB}=\Vp A\tmin\Vp B$, every state of $AB$ is a mixture of product states: the composite admits no entanglement. If $\Vp{AB}=\Vp A\tmax\Vp B$, the states of $AB$ are all normalised vectors that are nonnegative on all product effects, and the effect cone is $(\Vp A)^*\tmin(\Vp B)^*$: every effect is a nonnegative combination of product effects. For two gbits the maximal tensor product is the set of all no-signalling boxes with binary inputs and outputs, i.e., the conditional distributions $p(a,b\,|\,x,y)$ of the outcomes of the measurements $X$ or $Z$ on each gbit whose marginals for one gbit do not depend on the measurement chosen for the other \cite{Bar07,BLMPPR05}, including the Popescu--Rohrlich box $p(a,b\,|\,x,y)=\tfrac12[a\oplus b=xy]$ \cite{PR94}. For two qubits the quantum cone $\Pos_4$ lies strictly between the two: the maximally entangled state $\Phi^+=\tfrac12\sum_{i,j}\ket{ii}\bra{jj}$ is not a mixture of product states, and its partial transpose, which has the eigenvalue $-\tfrac12$, is nonnegative on all product effects and hence lies in $\Pos_2\tmax\Pos_2$ but not in $\Pos_4$. The two tensor products coincide if and only if one of the cones is a simplex cone \cite{ALPP21}.
\end{example}

\begin{fact}[Composites with a classical factor are unique \cite{Bar07,Pla23}]\label{fact:classical-factor}
Let $\Delta_n$ be classical and $B$ any system. Identifying $\R^n\otimes V_B$ with $V_B^{\oplus n}$ via $\delta_i\otimes\beta\mapsto(0,\dots,\beta,\dots,0)$,
\[
\R^n_{\ge0}\tmin\Vp B=\R^n_{\ge0}\tmax\Vp B=(\Vp B)^{\oplus n}=\Big\{\sum_i\delta_i\otimes\beta_i:\beta_i\in\Vp B\Big\}.
\]
Consequently every composite of $\Delta_n$ with $B$ has the cone $(\Vp B)^{\oplus n}$, and its states are the vectors $\sum_i\delta_i\otimes\beta_i$ with $\beta_i\in\Vp B$ and $\sum_iu_B(\beta_i)=1$: a classical value $i$ occurring with probability $p_i=u_B(\beta_i)$, together with the conditional state $\beta_i/p_i$ of $B$.
\end{fact}
\begin{proof}
Every $\alpha\in\R^n_{\ge0}$ is $\sum_i\alpha_i\delta_i$, so $\alpha\otimes\beta=\sum_i\delta_i\otimes(\alpha_i\beta)\in(\Vp B)^{\oplus n}$; conversely $\delta_i\otimes\beta_i$ is a product of cone elements. For the maximal tensor product write $\xi=\sum_i\delta_i\otimes\beta_i$ uniquely; the coordinate functionals $b_i$ lie in $(\R^n_{\ge0})^*$, and $(b_i\otimes g)(\xi)=g(\beta_i)\ge0$ for all $g\in(\Vp B)^*$ gives $\beta_i\in\Vp B$ by the bipolar theorem.
\end{proof}

In quantum theory this fact says that the joint states of a classical register and a quantum register are exactly the classical--quantum states $\sum_i\ketbra{i}\otimes\beta_i$. In particular, for classical $A$ and $B$ the composite is $\Delta_{nm}$.

\subsection{Channels and instruments}

\begin{definition}[Channel, instrument]\label{def:channel}
Fix a collection $\mathcal S$ of systems closed under composition. A linear map $\Phi:V_A\to V_B$ is \emph{positive on all extensions} if $(\Phi\otimes\id_{V_C})(\Vp{AC})\subseteq\Vp{BC}$ for every $C\in\mathcal S$, including $C=I$. A \emph{channel} from $A$ to $B$ is a map positive on all extensions with $u_B\circ\Phi=u_A$. An \emph{instrument} from $A$ to $B$ with outcomes $\{1,\dots,k\}$ is a family $(\Phi_1,\dots,\Phi_k)$ of maps positive on all extensions whose sum is a channel; on $\omega\in\Om A$ it produces the outcome $i$ with probability $u_B(\Phi_i(\omega))$ and leaves $B$ in the state $\Phi_i(\omega)/u_B(\Phi_i(\omega))$. Measurements are instruments with trivial output system; the preparation of $\omega\in\Om A$ is the channel $\R\to V_A$, $1\mapsto\omega$; the discard channel of $A$ is $u_A:V_A\to\R$. Channels compose sequentially and in parallel.
\end{definition}

In quantum theory the maps positive on all extensions are exactly the completely positive maps (Example~\ref{ex:transpose}), and the channels are the completely positive trace-preserving maps. The following consequence of unit preservation is the reason that interactive protocols are well defined: what happens later does not change the statistics of what was recorded earlier.

\begin{fact}[No signalling from the future \cite{CDP10}]\label{fact:causal}
Let $(\Phi_i)_i$ be an instrument from $A$ to $B$, $\Psi$ a channel from $B$ to $C$, and $\omega\in\Om A$. The probability of the outcome $i$ is the same whether or not $\Psi$ is applied afterwards. More generally, in a sequence of instruments the joint distribution of the first $j$ outcomes does not depend on the instruments applied later.
\end{fact}
\begin{proof}
$u_C\circ\Psi\circ\Phi_i=u_B\circ\Phi_i$. For a sequence of instruments, sum the later instruments over their outcomes, which gives channels, and iterate.
\end{proof}

\begin{fact}[Channels with classical input and output registers \cite{Bar07,Pla23}]\label{fact:classical-io}
Let $\Phi$ be a channel from $\Delta_n\otimes P$ to $\Delta_k\otimes P'$. There are maps $\Phi_{ij}:V_P\to V_{P'}$ ($1\le i\le n$, $1\le j\le k$), positive on all extensions, with $\sum_ju_{P'}\circ\Phi_{ij}=u_P$ for each $i$, such that $\Phi(\delta_i\otimes\tau)=\sum_j\delta_j\otimes\Phi_{ij}(\tau)$. If $k=1$, then $\Phi$ is the classically controlled channel that applies $\Phi_i:=\Phi_{i1}$ on the input $i$.
\end{fact}
\begin{proof}
Define $\Phi_{ij}(\tau)$ as the $j$-th component of $\Phi(\delta_i\otimes\tau)$ in $V_{P'}^{\oplus k}$. By Fact~\ref{fact:classical-factor}, $\Phi$ maps $(\Vp P)^{\oplus n}$ into $(\Vp{P'})^{\oplus k}$, so $\Phi_{ij}(\Vp P)\subseteq\Vp{P'}$; the same argument with an extension $C$, whose composite with a classical system is again a direct sum, gives positivity on all extensions. Unit preservation evaluated on $\delta_i\otimes\tau$ gives $\sum_ju_{P'}(\Phi_{ij}(\tau))=u_P(\tau)$.
\end{proof}

\begin{fact}[Positivity for minimal and maximal tensor products \cite{Pla23}]\label{fact:minmax-pos}
Let $\mathcal S$ be a collection of systems closed under composition in which every composite is a minimal tensor product, or in which every composite is a maximal tensor product. Then every linear map $\Phi:V_A\to V_B$ with $\Phi(\Vp A)\subseteq\Vp B$ is positive on all extensions.
\end{fact}
\begin{proof}
Minimal case: $\Vp{AC}$ is generated by the products $\alpha\otimes\gamma$, and $(\Phi\otimes\id)(\alpha\otimes\gamma)=\Phi(\alpha)\otimes\gamma\in\Vp B\tmin\Vp C$. Maximal case: for $\xi\in\Vp A\tmax\Vp C$, $f\in(\Vp B)^*$ and $g\in(\Vp C)^*$, $(f\otimes g)\big((\Phi\otimes\id)\xi\big)=\big((f\circ\Phi)\otimes g\big)(\xi)\ge0$ because $f\circ\Phi\in(\Vp A)^*$.
\end{proof}

\begin{example}[Positivity depends on the composites]\label{ex:transpose}
In classical theory every stochastic matrix is a channel, and every positive unit-preserving map is stochastic. In quantum theory a linear map $\Phi$ on $\Herm_d(\C)$ is positive on all extensions if and only if $\Phi\otimes\id_{\Herm_{d'}}$ maps $\Pos_{dd'}$ into $\Pos_{dd'}$ for all $d'$, i.e., if and only if $\Phi$ is completely positive; and every completely positive map is positive on all extensions, since its tensor product with an identity is again completely positive. The transposition map $\rho\mapsto\rho^{\mathsf T}$ on a qubit and the map $R(\rho)=\Tr(\rho)\one-\rho$ are positive and trace preserving, but not completely positive: $(\id\otimes{\mathsf T})(\Phi^+)$ and $(\id\otimes R)(\Phi^+)=\tfrac12\one\otimes\one-\Phi^+$ both have the eigenvalue $-\tfrac12$. So neither is a channel of quantum theory. In a theory whose composites of quantum systems are minimal tensor products, both are channels by Fact~\ref{fact:minmax-pos}. Example~\ref{ex:unot} uses $R$ to show that the value of a verifier can depend on this choice.
\end{example}

\subsection{Theories and the standing assumptions}\label{sec:standing}

\begin{definition}[Theory]\label{def:theory}
A theory $\GG$ consists of (a) a collection $\mathcal S_\GG$ of systems containing a finite set of elementary systems and closed under forming composites, the composition being associative and symmetric up to the obvious isomorphisms; and (b) for each pair $A,B\in\mathcal S_\GG$ a set $\Chan AB$ of allowed channels and a set of allowed instruments, all satisfying Definition~\ref{def:channel} with respect to $\mathcal S_\GG$, closed under sequential and parallel composition, containing the identity channels, the discard channels $u_A$ and the preparations of all states of all systems, and such that outcomes of instruments may be coarse-grained and every allowed instrument $(\Phi_i)_i$ has an allowed flagged channel $\omega\mapsto\sum_i\delta_i\otimes\Phi_i(\omega)$ that records the outcome in a classical register. A theory $\GG$ is a \emph{subtheory} of $\GG'$ if every system, allowed channel and allowed instrument of $\GG$ is one of $\GG'$.
\end{definition}

Throughout the paper, every theory satisfies the following standing assumptions.
\begin{enumerate}[label=\textbf{(T\arabic*)},leftmargin=3.2em]
\item \textbf{Local tomography.} Composites are as in Definition~\ref{def:composite}: $V_{AB}=V_A\otimes V_B$.
\item \textbf{Causality.} Allowed instruments are normalised: $\sum_i\Phi_i$ preserves the unit effect. This is built into Definition~\ref{def:channel}; we list it separately because it is a physical assumption, equivalent to the uniqueness of the deterministic effect \cite{CDP10}, whose consequence Fact~\ref{fact:causal} is what makes interactive protocols well defined.
\item \textbf{Classical theory is contained.} Every standard $\Delta_n$, with the coordinates of Example~\ref{ex:classical}, is a designated classical system of $\GG$; every stochastic map between classical systems is an allowed channel; the readout effects of $\Delta_n$ are allowed; and classical control is allowed: if $\Phi_1,\dots,\Phi_n\in\Chan AB$, then $\delta_i\otimes\alpha\mapsto\delta_i\otimes\Phi_i(\alpha)$ is an allowed channel from $\Delta_nA$ to $\Delta_nB$, and likewise for instruments.
\item \textbf{Computable gates.} The elementary systems have bounded dimension, and $\GG$ comes with a finite gate set $G$ of allowed instruments, each acting on a composite of at most a constant number of elementary systems, such that in a fixed basis of each elementary $V_A$ every matrix entry of every outcome map of every gate, and every coordinate of every elementary unit effect $u_A$, can be approximated to within $2^{-k}$ by a deterministic algorithm running in time $\poly(k)$. The gate set contains the preparation of a fixed state of each elementary system, the elementary discard channels, the readout of classical bits, the preparation of a fair coin, and a universal set of classical reversible gates \cite{Tof80}. The gate set is part of the specification of $\GG$; no universality for all allowed transformations is assumed.
\end{enumerate}

{Assumption (T4) constrains the verifier's gates and nothing else.} It does not require the state cones of the systems to be computable, and Section~\ref{sec:beyond} shows that this omission matters. {The classes defined below depend on $\GG$ together with its designated gate set; the coordinates determine whether (T4) holds.} We call the coordinates of the elementary systems together with the gate set the \emph{presentation} of $\GG$, and we write $\GIP(\GG)$ with the presentation understood. An \emph{isomorphism} of theories is a family of linear order isomorphisms $\iota_A:V_A\to V_{A'}$ with $u_{A'}\circ\iota_A=u_A$, compatible with composition and mapping allowed channels and instruments onto allowed ones. Isomorphic theories have the same operational content: the same statistics for all preparations, transformations and measurements. An isomorphism need not preserve the computability of gate entries. {Transporting the designated gates together with the states and transformations preserves circuit descriptions and acceptance probabilities; replacing them by a standard gate set can change the class.}

\begin{example}[Classical theory]\label{ex:cl-theory}
Systems $\Delta_n$; composites $\Delta_n\Delta_m=\Delta_{nm}$, forced by Fact~\ref{fact:classical-factor}; channels the column-stochastic matrices; gate set the Toffoli gate, fair coin, copy, bit readout and discard. Polynomial-time uniform circuits of this theory are randomised Boolean circuits, and the bounded-error class is $\BPP$.
\end{example}

\begin{example}[Quantum theory with classical systems, $\QC$]\label{ex:q-theory}
Elementary systems: the qubit $Q$ and the classical bit $\Delta_2$ (we sometimes adjoin further classical systems $\Delta_d$ of fixed dimension as elementary systems). Systems: the composites of elementary systems, which are the classical--quantum registers $\bigoplus_{c}\Herm(\C^{2^a})$ with a classical label $c$ and $a$ qubits, together with quantum systems $Q_d$ and classical systems $\Delta_n$ of all dimensions and their composites. Composites of quantum systems are $\Herm(\C^d\otimes\C^{d'})$ with the positive semidefinite cone, and composites with classical systems are direct sums; (T1) holds because $\Herm(\C^d\otimes\C^{d'})=\Herm_d\otimes\Herm_{d'}$ as real vector spaces. Channels are the completely positive trace-preserving (CPTP) maps and instruments are the families of completely positive maps with trace-preserving sum. Gate set: Hadamard, $\pi/8$, controlled-NOT, preparation of $\ket0$, measurement in the computational basis, and the classical gates. Polynomial-time uniform circuits of $\QC$ are polynomial-time quantum circuits with intermediate measurements and classical control \cite{AKN98}, and the bounded-error class is $\BQP$. We write ``quantum theory'' for $\QC$ when no confusion can arise.
\end{example}

\begin{example}[Boxworld]\label{ex:boxworld}
Elementary systems: the gbit and the classical bit; systems: all composites of gbits with the standard classical systems $\Delta_n$. All composites are maximal tensor products, so the states of $n$ gbits are the no-signalling boxes with $n$ binary inputs and $n$ binary outputs. Channels: all positive unit-preserving maps, which by Fact~\ref{fact:minmax-pos} are exactly the maps positive on all extensions; instruments: all families of positive maps with unit-preserving sum. Boxworld \cite{Bar07} contains the Popescu--Rohrlich box and violates the CHSH inequality \cite{CHSH69} maximally; its reversible transformations are generated by relabellings and permutations \cite{GMCD10}. Gate set: the measurements $X$ and $Z$, preparation of the pure gbit states and of the PR box, discard, and the classical gates. (T1)--(T4) hold, and every gate has entries in $\{0,\pm1,\pm\tfrac12\}$.
\end{example}

\begin{example}[Quantum systems with minimal composites, $\GG_{\min}$]\label{ex:Gmin}
Elementary systems: the qubit and the classical bit; systems: all composites of qubits with the standard classical systems $\Delta_n$, with all composites minimal tensor products, all positive unit-preserving maps as channels and all families of positive maps with unit-preserving sum as instruments (Fact~\ref{fact:minmax-pos}), and the gate set consisting of the Hadamard and $\pi/8$ gates, preparation of $\ket0$, measurement in the computational basis and the classical gates. This is a theory satisfying (T1)--(T4), with the same elementary systems and the same single-qubit gates as $\QC$. It has no entangled states, and the controlled-NOT gate is not a channel of $\GG_{\min}$: it maps the product state $\ketbra+\otimes\ketbra0$ to a state outside $\Pos_2\tmin\Pos_2$. On the other hand, the maps $\mathsf T$ and $R$ of Example~\ref{ex:transpose} are channels of $\GG_{\min}$.
\end{example}

\begin{example}[Real quantum theory: a theory violating (T1)]\label{ex:realQ}
Real quantum theory has systems $\Sym_d(\R)$, the real symmetric $d\times d$ matrices with the positive semidefinite cone and the trace, and composites $\Sym_{dd'}(\R)$ on $\R^d\otimes\R^{d'}$ \cite{Stu60,HW12}. This violates (T1): $\dim\Sym_{dd'}=\tfrac12dd'(dd'+1)$ exceeds $\dim(\Sym_d\otimes\Sym_{d'})=\tfrac14dd'(d+1)(d'+1)$ when $d,d'>1$. Concretely, for two rebits (real two-level systems, $\Sym_2(\R)$) the states $\tfrac14(\one\pm Y\otimes Y)$, where $Y=\left(\begin{smallmatrix}0&-i\\i&0\end{smallmatrix}\right)$ so that $Y\otimes Y$ is real, agree on every product of real observables but are different states. Real quantum theory is a consistent theory, experimentally distinguishable from complex quantum theory \cite{RTW21}; we exclude it here by (T1) and return to it in Section~\ref{sec:discussion}.
\end{example}

\begin{example}[Classical theory with a triangular system adjoined, $\GG_v$]\label{ex:Gv}
Fix $v\in(0,1)$. Elementary systems: the triangular system $W_v$ of Example~\ref{ex:triangle} and the classical bit. Systems: all finite standard classical systems $\Delta_n$ and their composites with copies of $W_v$, every composite being a minimal tensor product; this is associative and symmetric. Channels: all positive unit-preserving maps, which by Fact~\ref{fact:minmax-pos} are positive on all extensions; instruments: families of positive maps with unit-preserving sum. Gate set: the preparation of $e_0$, the discard $u$, the measurement $(f_y,u-f_y)$, and the classical gates required by (T4). Then (T1)--(T3) hold by construction, and (T4) holds because every gate matrix has entries in $\{0,\pm1,\tfrac12\}$ and the unit effect is $(0,0,1)$. The real number $v$ appears in no gate: it appears only in the description of the state cone, which no standing assumption constrains. The states of $W_v^{\otimes k}$ are the mixtures of the products $e_{i_1}\otimes\dots\otimes e_{i_k}$, and the product effects $f_y^{(i)}=u\otimes\dots\otimes f_y\otimes\dots\otimes u$ are effects of $W_v^{\otimes k}$. The maps $L_v^{\otimes k}\otimes\id$ constitute an isomorphism of $\GG_v$ with classical theory having $\Delta_3$ as an additional elementary system. This isomorphism is computable if and only if $v$ is computable: the matrix of $L_v^{-1}$ has the columns $e_0,e_1,e_2$, and that of $L_v$ involves $1/v$.
\end{example}

\begin{example}[Quantum theory with a triangular system adjoined, $\GG^{\mathrm Q}_v$]\label{ex:GQv}
Fix $v\in(0,1)$. Elementary systems: the qubit, the classical bit and $W_v$. Systems: $S=A\otimes W_v^{\otimes c}$, where $A$ is a system of $\QC$ and $c\ge0$; we define their cones, unit effects and composites by transporting those of the $\QC$-systems $A\otimes\Delta_3^{\otimes c}$ through $\id_{V_A}\otimes L_v^{\otimes c}$. Allowed channels and instruments are the transports of the CPTP maps and of the completely positive instruments of $\QC$. Gate set: the quantum and classical gates of Example~\ref{ex:q-theory} together with the preparation of $e_0$, the discard $u$ and the measurement $(f_y,u-f_y)$ on $W_v$. By construction $\GG^{\mathrm Q}_v$ is isomorphic to $\QC$ with $\Delta_3$ adjoined as an elementary system, so (T1)--(T3) hold; (T4) holds because every gate entry is an entry of a quantum gate or lies in $\{0,\pm1,\tfrac12\}$. The theory contains $\QC$, and hence classical theory, as a subtheory with its standard gate set. In the transported coordinates the measurement $(f_y,u-f_y)$ has the effects $(0,0,v)$ and $(1,1,1-v)$ on $\Delta_3$: it can be realised by reading the label and, if the label is $2$, tossing a coin with bias $v$. The isomorphism with $\QC$ is computable if and only if $v$ is.
\end{example}

\begin{example}[Quantum theory with an additional system and maximal composites, $\GG_E$]\label{ex:foil}
Let $E$ be a system with $V_E=\Herm_2(\C)$, $\Vp E=\Pos_2$ and $u_E=\Tr$; as a single system, $E$ is a qubit. Elementary systems: $Q$, $\Delta_2$ and $E$; systems: all composites of these with the standard classical systems $\Delta_n$. For a register with $a$ qubits and $b$ copies of $E$ (classical factors are direct sums, by Fact~\ref{fact:classical-factor}) define the cone
\[
C_{a,b}:=\big(\Pos_{2^a}\tmin\Pos_2\tmin\dots\tmin\Pos_2\big)^*\qquad(b\text{ copies of }\Pos_2),
\]
the maximal tensor product of the quantum block $\Pos_{2^a}$ with the $b$ additional systems. Then $C_{a,0}=\Pos_{2^a}$ by self-duality of the positive semidefinite cone, so composites of qubits are ordinary quantum composites, while $C_{1,1}$ consists of the block-positive operators on $\C^2\otimes\C^2$, i.e., the Hermitian operators $\xi$ with $\bra{a\otimes b}\xi\ket{a\otimes b}\ge0$ for all product vectors. These include $\tfrac12\mathrm{SWAP}$: it is normalised, block-positive since $\bra{a\otimes b}\mathrm{SWAP}\ket{a\otimes b}=|\langle a|b\rangle|^2$, and has the eigenvalue $-\tfrac12$ on the singlet. Allowed channels: all unit-preserving maps positive on all extensions; allowed instruments: all families of maps positive on all extensions with unit-preserving sum; gate set: the quantum and classical gates together with the preparation of $\one/2$ and the discard on $E$. We verify that this is a theory. \emph{Products.} If $\xi\in C_{a,b}$ and $\xi'\in C_{a',b'}$, then for $f\in\Pos_{2^{a+a'}}$ and a product $g=g_1\otimes g_2$ of positive elements on the $b$ and $b'$ additional systems, $(f\otimes g)(\xi\otimes\xi')=f(X\otimes X')$ with $X=(\id\otimes g_1)(\xi)$ and $X'=(\id\otimes g_2)(\xi')$; $X\in\Pos_{2^a}$ because $\xi$ is nonnegative on all products of positive elements, and likewise $X'\in\Pos_{2^{a'}}$, so $(f\otimes g)(\xi\otimes\xi')\ge0$ and $\xi\otimes\xi'\in C_{a+a',b+b'}$. \emph{Inclusion in the maximal tensor product.} The dual of $C_{a,b}$ is the minimal tensor product $\Pos_{2^a}\tmin\Pos_2\tmin\dots$ (bipolar theorem), so a product of effects of $C_{a,b}$ and $C_{a',b'}$ is a sum of products of positive elements, which is nonnegative on $C_{a+a',b+b'}$; hence $C_{a+a',b+b'}\subseteq C_{a,b}\tmax C_{a',b'}$. Associativity and symmetry hold because $C_{a,b}$ depends only on $(a,b)$. \emph{Quantum channels are allowed.} If $\Phi$ is a completely positive map between qubit registers, then for $\xi\in C_{a+a'',b}$, $f\in\Pos_{2^{a'+a''}}$ and a product $g$, $(f\otimes g)\big((\Phi\otimes\id)\xi\big)=\big(((\Phi^*\otimes\id)f)\otimes g\big)(\xi)\ge0$ because $\Phi^*\otimes\id$ is positive; so $\Phi$ is positive on all extensions, and the same argument shows that every completely positive instrument between qubit registers is an allowed instrument. Thus $\GG_E$ satisfies (T1)--(T4) and contains quantum theory in the sense of Definition~\ref{def:containsQ} below. A prover of $\GG_E$ can prepare $\tfrac12\mathrm{SWAP}$ on $Q\otimes E$, send $Q$ to the verifier and keep $E$, so that its memory is correlated with the verifier's qubit through a state outside the quantum cone. Theorem~\ref{thm:qm-universal} shows that this gives it no advantage against a quantum verifier, and Example~\ref{ex:foil-steering} shows why.
\end{example}

\section{Circuits and path sums}\label{sec:circuits}

A circuit in a theory $\GG$ is a finite acyclic arrangement of gates from the gate set $G$ on wires labelled by systems; every gate outcome is recorded in a classical register. Circuits may be adaptive: by (T3) the gate applied at a given step may be chosen as a function of the outcomes recorded so far. Formally, an adaptive circuit is a deterministic procedure that, given the outcomes recorded so far, either halts or names the next gate and the wires it acts on. Circuits may take classical inputs $a$, fed into classical wires; this is how a verifier receives the prover's classical messages.

Let $C$ be an adaptive circuit with $L$ gate applications, $a$ a classical input, and $z=(z_1,\dots,z_L)$ an outcome record. Let $V_j$ be the vector space of the whole register after gate $j$, of dimension $D_j\le D$. By (T1) the outcome map of gate $j$, tensored with the identity on the untouched wires, is a matrix $M_j:V_{j-1}\to V_j$ in the product bases (rectangular when systems are prepared or discarded). The record $z$ occurs with probability
\begin{equation}\label{eq:record}
\Pr[z\,|\,a]=u\big(M_L\cdots M_2M_1\omega_0\big),
\end{equation}
where $\omega_0\in V_0$ is the initial product state, including the classical input, and $u$ is the unit effect of $V_L$. By (T2), $\sum_z\Pr[z\,|\,a]=1$.

\begin{definition}[Uniform circuit families \cite{LB15}]\label{def:uniform}
A family $\{C_x\}_{x\in\{0,1\}^*}$ of adaptive circuits in $\GG$ is polynomial-time uniform if its gates come from $G$, the numbers of gate applications and of wires in $C_x$ are at most $p(|x|)$ for a polynomial $p$, and a polynomial-time Turing machine, given $x$ and the outcomes recorded so far, outputs the next gate and its wires or halts. An acceptance rule is a polynomial-time computable Boolean function $\acc(x,z)$ of the input and the record.
\end{definition}

Lee and Barrett define $\BGP(\GG)$ as the class of promise problems decided with bounded error by such families \cite{LB15}; it equals $\BPP$ for classical theory and $\BQP$ for quantum theory, and for every theory satisfying (T1)--(T4) it is contained in $\AWPP\subseteq\PP$ \cite{LB15}. We use only the following ingredient of their argument.

The matrices $M_j$ in \eqref{eq:record} are exponentially large, but their entries are efficiently computable: by (T1) each $M_j$ is, up to a permutation of tensor factors, the Kronecker product of a constant-size gate matrix with identity matrices, so $(M_j)_{k,l}$ is either zero or a single gate entry, and which one is decided in polynomial time from the binary encodings of $k$ and $l$. Expanding the product gives a sum over paths,
\begin{equation}\label{eq:pathsum}
\Pr[z\,|\,a]=\sum_{k_0,\dots,k_L}u_{k_L}(M_L)_{k_Lk_{L-1}}\cdots(M_1)_{k_1k_0}(\omega_0)_{k_0},
\end{equation}
where $u_{k}$ and $(\omega_0)_k$ are products, over the wires, of coordinates of elementary unit effects and elementary prepared states. For a randomised Boolean circuit this is the sum over computation paths of their probabilities; for a quantum circuit written in the density-matrix picture it is the familiar sum over pairs of Feynman paths that underlies $\BQP\subseteq\PSPACE$ \cite{BV97,ADH97}.

\begin{fact}[Path-sum representation \cite{LB15}]\label{fact:pathsum}
Let $\{C_x\}$ be a polynomial-time uniform family in a theory $\GG$. For every polynomial $q$ there is a deterministic polynomial-space algorithm that, given $x$, a classical input $a$ and a record $z$ of length $\poly(|x|)$, outputs a rational number $\tilde P$ with $|\tilde P-\Pr[z\,|\,a]|\le2^{-q(|x|)}$. The same holds for the sums $\sum_{z\in Z}\Pr[z\,|\,a]$ over any polynomial-time decidable set $Z$ of records.
\end{fact}
\begin{proof}
Let $n=|x|$, $L\le p(n)$ and $D\le2^{p'(n)}$, and let $2^{b(n)}$ bound the absolute value of every gate entry, every elementary unit-effect coordinate and every coordinate of an elementary prepared state. Each term of \eqref{eq:pathsum} is a product of at most $\ell=\poly(n)$ factors (the $L$ matrix entries and the polynomially many elementary coordinates that make up $u_{k_L}$ and $(\omega_0)_{k_0}$; we pad shorter products with factors equal to $1$, so that $\ell$ does not depend on the record), and there are $2^{s(n)}$ terms for a polynomial $s$, including the outer sum over records when required. By (T4) each factor can be computed on a common dyadic grid to error $\eps=2^{-t(n)}$. Replacing the factors of a product one at a time (telescoping) bounds the error of each term by $\ell\,2^{(b(n)+1)(\ell-1)}\eps$, so it suffices to take $t(n)\ge q(n)+s(n)+(b(n)+1)(\ell-1)+\lceil\log_2\ell\rceil+2$. Enumerating paths and records in lexicographic order, every rounded term has the common denominator $2^{t(n)\ell}$, the products and the accumulator have polynomially many bits, and the computation uses polynomial space.
\end{proof}

\section{Interactive proofs in a theory}\label{sec:model}

\subsection{Verifiers and prover strategies}

An interactive protocol alternates messages between a verifier and a prover. We organise it in rounds: in round $i$ the verifier sends a system $Q_i$ and the prover returns a system $A_i$ (Figure~\ref{fig:protocol}). A message system may be trivial, so this pattern includes protocols in which the prover speaks first and protocols with an odd number of messages. After the last round the verifier computes a bit and announces it. We count a message whenever a nontrivial system is transferred; the verifier's final measurement is not a message.

\begin{figure}[t]
\centering
\begin{tikzpicture}[x=1cm,y=1cm,
  box/.style={draw,rounded corners=2pt,minimum height=0.8cm,minimum width=1.15cm,fill=gray!10},
  msg/.style={-{Stealth[length=2mm]},thick},
  mem/.style={thick}]
  \node[box] (V1) at (0,2) {$V^{(1)}$};
  \node[box] (V2) at (3.2,2) {$V^{(2)}$};
  \node[box] (V3) at (6.4,2) {$V^{(3)}$};
  {\node at (7.15,2) {$\cdots$};}
  {\node[box,draw=blue] (Vr) at (7.9,2) {$V^{(r)}$};}
  \node[box] (Vf) at (10.6,2) {$V^{(r+1)}$};
  \node[box] (P1) at (1.6,0) {$\Pi_1$};
  \node[box] (P2) at (4.8,0) {$\Pi_2$};
  \node at (7.4,0) {$\cdots$};
  \node[box] (Pr) at (9.0,0) {$\Pi_r$};
  \draw[mem] (-1.3,2) node[left]{$\omega_0$} -- (V1);
  \draw[mem] (V1) -- node[above]{$V_1^-$} (V2);
  \draw[mem] (V2) -- node[above]{$V_2^-$} (V3);
  {\draw[mem,draw=blue] (V3) -- (7.0,2);}
  {\draw[mem,draw=blue] (Vr) -- (Vf);}
  \draw[mem] (Vf) -- (12.2,2) node[right]{$\acc$};
  \draw[mem] (0.2,0) node[left]{$\tau_0$} -- (P1);
  \draw[mem] (P1) -- node[below]{$P_1$} (P2);
  \draw[mem] (P2) -- (6.8,0);
  \draw[mem] (8.0,0) -- (Pr);
  \draw[mem] (Pr) -- (10.6,0) node[right]{$P_r$};
  \draw[msg] (V1.south east) -- node[left,pos=0.5]{$Q_1$} (P1.north west);
  \draw[msg] (P1.north east) -- node[right,pos=0.5]{$A_1$} (V2.south west);
  \draw[msg] (V2.south east) -- node[left,pos=0.5]{$Q_2$} (P2.north west);
  \draw[msg] (P2.north east) -- node[right,pos=0.5]{$A_2$} (V3.south west);
  {\draw[msg,draw=blue] (Vr.south east) -- node[left,pos=0.5]{$Q_r$} (Pr.north west);}
  \draw[msg] (Pr.north east) -- node[right,pos=0.5]{$A_r$} (Vf.south west);
\end{tikzpicture}
\caption{An $r$-round protocol. The verifier's stages $V^{(i)}$ are adaptive circuits of $\GG$ acting on its private register (top wire); the prover's channels $\Pi_i$ act on its private system (bottom wire). The verifier and the prover start in independently prepared states $\omega_0$ and $\tau_0$ and interact only through the message systems $Q_i$ and $A_i$.}
\label{fig:protocol}
\end{figure}
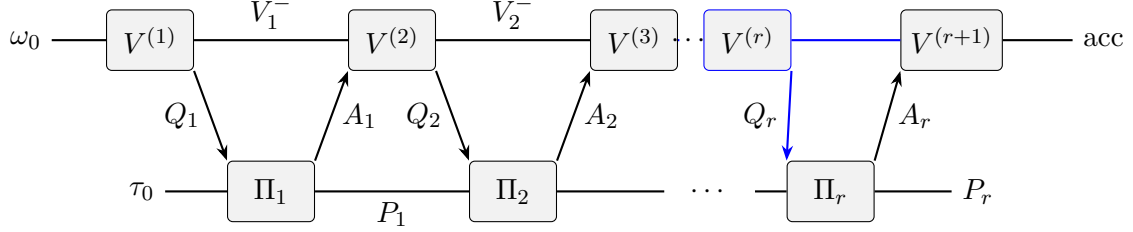

\begin{definition}[Verifier]\label{def:verifier}
Let $\GG$ be a theory and $r=r(n)$ a polynomial. An $r$-round verifier $V$ for $\GG$ is a polynomial-time uniform family $\{V_x\}$ of adaptive circuits with the following structure. Before round $i$ the verifier holds a private register $V_{i-1}$, which includes all messages retained from earlier rounds; its stage-$i$ circuit $V^{(i)}$ acts on this register and on fresh ancillas, produces an outgoing message system $Q_i$ and a retained register $V_i^-$, and hands $Q_i$ to the prover; on receiving $A_i$ its register becomes $V_i=V_i^-A_i$. In stage $r+1$ it produces a classical bit $\acc$ among its recorded outcomes. All registers contain at most $\poly(n)$ elementary systems, the initial private state is prepared by verifier gates, the message types are fixed by $x$ (with padding), and the gates and wires are determined by $x$ and the recorded outcomes via the uniformity machine.
\end{definition}

\begin{definition}[Prover]\label{def:prover}
A prover strategy $P$ for an $r$-round verifier in $\GG$ consists of a private system $P_0$ of $\GG$ with a state $\tau_0\in\Om{P_0}$, systems $P_1,\dots,P_r$ of $\GG$, and channels $\Pi_i\in\Chan{Q_iP_{i-1}}{A_iP_i}$ for $i=1,\dots,r$. There is no bound on the dimensions of the $P_i$, and any allowed channel may be used.
\end{definition}

We could also allow the prover to use instruments and retain their classical outcomes; this does not change the model, because an instrument followed by a classically controlled continuation is a channel on a larger private system. In classical theory a prover strategy is exactly a randomised prover with memory; in quantum theory it is exactly a quantum prover in the sense of Kitaev and Watrous \cite{KW00}, by Stinespring dilation \cite{Sti55}.

The verifier and the prover begin with independently prepared states. Given $V_x$ and $P$, the interaction is a closed circuit in $\GG$, and $\acc$ has a well-defined probability of being $1$, which we denote by $\Pr[V_x\leftrightarrow P\text{ accepts}]$; by (T2) and Fact~\ref{fact:causal}, the distribution of the verifier's recorded outcomes at each stage does not depend on operations performed later.

\subsection{Values and classes}

\begin{definition}[Value; the classes]\label{def:classes}
The value of $V$ on $x$ is $\val_\GG(V_x)=\sup_P\Pr[V_x\leftrightarrow P\text{ accepts}]$, the supremum over all prover strategies in $\GG$. A promise problem $L$ is in $\GIP(\GG)$ if there is a verifier $V$ for $\GG$ with $\val_\GG(V_x)\ge2/3$ for $x\in L_{\mathrm{yes}}$ and $\val_\GG(V_x)\le1/3$ for $x\in L_{\mathrm{no}}$. Variants:
\begin{itemize}
\item $\GIP(\GG;k)$: at most $k$ nontrivial messages; $\GMA(\GG)=\GIP(\GG;1)$ with the message sent by the prover.
\item $\GIPcl(\GG)$ and $\GIPcl(\GG;k)$: all message systems $Q_i,A_i$ are designated standard classical registers from (T3), encoded by polynomially many bits with the efficient readout of (T4); the verifier and the prover may use arbitrary systems of $\GG$ internally.
\item $\GIPqm(\GG)$ and $\GIPqm(\GG;k)$, defined when $\GG$ contains quantum theory (Definition~\ref{def:containsQ}): the verifier is a quantum verifier in the sense of Definition~\ref{def:QIP}, whose registers are qubit registers together with classical registers for recorded outcomes and whose messages are qubit registers; the prover may use arbitrary systems of $\GG$.
\end{itemize}
By construction, $\GIP(\text{classical theory})=\GIPcl(\text{classical theory})=\IP$, and $\GIPqm(\QC)=\QIP$; Lemma~\ref{lem:QC} shows that also $\GIP(\QC)=\QIP$. The classes $\GMA$ specialise to $\MA$ and $\QMA$.
\end{definition}

\begin{definition}[Theories containing quantum theory]\label{def:containsQ}
A theory $\GG$ contains quantum theory if (i) the qubit $Q$ is an elementary system of $\GG$ and, under the identification of (T1), the composite $Q^{\otimes n}$ in $\GG$ is the system $(\Herm(\C^{2^n}),\Pos_{2^n},\Tr)$ for every $n$; (ii) every CPTP map $Q^{\otimes n}\to Q^{\otimes m}$ is an allowed channel and every family of completely positive maps summing to a CPTP map is an allowed instrument; (iii) the gate set of $\GG$ contains the quantum gates of Example~\ref{ex:q-theory}. A quantum verifier for $\GG$ is then a verifier whose registers are qubit registers with classical registers for recorded outcomes and whose gates are the quantum and classical gates; it is exactly a $\QIP$ verifier.
\end{definition}

The theories $\QC$, $\GG^{\mathrm Q}_v$ and $\GG_E$ contain quantum theory; classical theory, boxworld, $\GG_{\min}$ and $\GG_v$ do not. The theory $\GG_{\min}$ has the qubit as an elementary system and allows all quantum channels on single qubits, but its composites of qubits are not the quantum composites, so condition (i) fails.

We stress that ``classical messages'' refers to the fixed effective encodings of (T3)--(T4). It does not include every simplicial state space merely because an abstract order isomorphism to a standard $\Delta_n$ exists: such an isomorphism need not preserve the computability of the verifier's gates, and Section~\ref{sec:beyond} exploits exactly this distinction.

\subsection{Basic properties and error reduction}

For a fixed verifier the value is monotone under inclusion of theories: if $\GG$ is a subtheory of $\GG'$ and the gate set of $\GG$ is contained in that of $\GG'$, then every $\GG$-verifier is a $\GG'$-verifier and its value can only increase. This does not by itself compare the classes, because a larger value preserves completeness but may violate soundness. The following lemma is the tool we use to compare classes.

\begin{lemma}[Simulation]\label{lem:simulation}
Let $\GG$ be a subtheory of $\GG'$ with gate set contained in that of $\GG'$. Suppose that for every verifier $V$ for $\GG$ and every prover strategy $P'$ in $\GG'$ there is a prover strategy $P$ in $\GG$ with $\Pr[V_x\leftrightarrow P\text{ accepts}]=\Pr[V_x\leftrightarrow P'\text{ accepts}]$ for every $x$. Then $\GIP(\GG)\subseteq\GIP(\GG')$, with the number of messages preserved, and $\GIP(\GG)=\GIP(\GG')$ if in addition every verifier for $\GG'$ is a verifier for $\GG$.
\end{lemma}
\begin{proof}
A $\GG$-verifier $V$ is a $\GG'$-verifier, and by hypothesis its value over $\GG'$-provers equals its value over $\GG$-provers, so the same $V$ witnesses $L\in\GIP(\GG')$. If every $\GG'$-verifier is a $\GG$-verifier, the same argument applies in the other direction.
\end{proof}

The constants $2/3$ and $1/3$ in Definition~\ref{def:classes} are not important, by the following error reduction. {A prover facing $k$ sequential repetitions may correlate its behaviour across the repetitions through its memory; causality prevents this from improving its success probability.}

\begin{fact}[Stochastic domination \cite{SS07}]\label{fact:domination}
Let $X_1,\dots,X_k$ be $\{0,1\}$-valued random variables with $\Pr[X_{j+1}=1\,|\,X_1,\dots,X_j]\le p$ almost surely for $0\le j<k$. Then $\Pr[\sum_jX_j\ge m]\le\Pr[\mathrm{Bin}(k,p)\ge m]$ for every $m$, where $\mathrm{Bin}(k,p)$ is a binomial random variable.
\end{fact}
\begin{proof}
Induction on $k$; the case $k=0$ is trivial. Write $\beta_m=\Pr[\mathrm{Bin}(k-1,p)\ge m]$, which is nonincreasing in $m$ with $\beta_m=1$ for $m\le0$. Conditioned on $X_1$, the remaining variables satisfy the hypothesis, so by induction $\Pr[\sum_{j\ge2}X_j\ge m-1\,|\,X_1=1]\le\beta_{m-1}$ and $\Pr[\sum_{j\ge2}X_j\ge m\,|\,X_1=0]\le\beta_m$. With $q=\Pr[X_1=1]\le p$,
\[
\Pr\Big[\sum_jX_j\ge m\Big]\le q\beta_{m-1}+(1-q)\beta_m\le p\beta_{m-1}+(1-p)\beta_m=\Pr[\mathrm{Bin}(k,p)\ge m],
\]
where the second inequality uses $\beta_{m-1}\ge\beta_m$.
\end{proof}

\begin{proposition}[Error reduction]\label{prop:error}
Let $V$ be an $r$-round verifier for $\GG$ and $k$ an odd polynomial. Let $V^{(k)}$ run $V_x$ $k$ times in sequence on freshly prepared private registers, discard the private register after each run, keep the $k$ accept bits $o_1,\dots,o_k$, and accept if and only if $\sum_jo_j>k/2$. Then
\[
\val(V_x)\ge2/3\implies\val(V^{(k)}_x)\ge1-e^{-k/20},\qquad
\val(V_x)\le1/3\implies\val(V^{(k)}_x)\le e^{-k/18}.
\]
Consequently the thresholds $(2/3,1/3)$ may be replaced by $(1-2^{-p(n)},2^{-p(n)})$ for any polynomial {$p$ with $p(n)\ge2$}, or by any constants {$0<s<c<1$}, without changing $\GIP(\GG)$, $\GIPcl(\GG)$ or $\GIPqm(\GG)$. Sequential repetition increases the number of messages, so it does not by itself amplify the variants with a fixed number of messages.
\end{proposition}
\begin{proof}
\emph{Completeness.} Choose a strategy $P$ with acceptance probability at least $0.66$ and let the prover run $k$ independent copies of $P$ on disjoint private systems. The closed circuit is then a parallel composition of $k$ independent closed circuits, so the accept bits are independent Bernoulli variables with success probability at least $0.66$, and Hoeffding's inequality \cite{Hoe63} gives $\Pr[\sum_jo_j\le k/2]\le e^{-2k(0.16)^2}\le e^{-k/20}$.

\emph{Soundness.} Let $P'$ be any strategy for $V^{(k)}_x$. We claim that for every $j<k$ and every $o\in\{0,1\}^j$ with $\Pr[o_1\cdots o_j=o]>0$,
\begin{equation}\label{eq:condsound}
\Pr[o_{j+1}=1\,|\,o_1\cdots o_j=o]\le\val(V_x).
\end{equation}
Just before run $j+1$ the verifier holds a classical register $R$ containing its record so far (discarding the earlier private registers does not affect any probability, by Fact~\ref{fact:causal}), a fresh private register in the fixed state $\omega_0$, and the prover holds its private system $P$. By Fact~\ref{fact:classical-factor} the joint state of $R$ and $P$ is $\sum_{o'}\delta_{o'}\otimes\tau_{o'}$ with $\tau_{o'}\in\Vp P$ and $u_P(\tau_{o'})=\Pr[o']$, where $o'$ ranges over the possible records, and the fresh register is in a product with it. Let $\tau_o=\sum_{o'}\tau_{o'}$, the sum over the records $o'$ whose accept bits are $o$, so that $u_P(\tau_o)=\Pr[o]$. Let $\mathrm{Acc}(\tau)$ be the probability that run $j+1$ accepts when the prover's system starts in $\tau\in\Om P$ and the prover applies the channels of $P'$ scheduled for that run; it is the restriction of a linear functional on $V_P$. Then $\Pr[o_1\cdots o_j=o,\ o_{j+1}=1]=\mathrm{Acc}(\tau_o)=\Pr[o]\,\mathrm{Acc}(\tau_o/\Pr[o])$, and $\tau_o/\Pr[o]\in\Om P$ together with the channels of run $j+1$ is a legitimate strategy for $V_x$, so $\mathrm{Acc}(\tau_o/\Pr[o])\le\val(V_x)$. This proves \eqref{eq:condsound}. Fact~\ref{fact:domination} with $p=1/3$ and Hoeffding's inequality give $\Pr[\sum_jo_j>k/2]\le\Pr[\mathrm{Bin}(k,1/3)\ge k/2]\le e^{-2k(1/6)^2}=e^{-k/18}$.

The statement about arbitrary constants {$0<s<c<1$} follows in the same way, with the threshold $k/2$ replaced by $k(c+s)/2$.
\end{proof}

\section{Classical messages}\label{sec:cl}

In this section $\GG$ is an arbitrary theory satisfying (T1)--(T4), and we prove that protocols with classical messages decide exactly $\PSPACE$ in every such theory. The proof has two parts. First, the prover's nonclassical resources are useless to it: its behaviour is reproduced by a classical randomised prover, and in fact by a deterministic one. Second, the value of the protocol is computable in polynomial space, {because the verifier's outcome probabilities are path sums.}

\subsection{Conditional product states}\label{sec:product}

Fix a verifier $V$ all of whose message systems are classical, and a prover strategy $P$ with initial state $\tau_0$ and channels $\Pi_1,\dots,\Pi_r$. Reading a classical register without disturbing it is the instrument whose $q$-th outcome map is the projection $\sum_{q'}\delta_{q'}\otimes\beta_{q'}\mapsto\delta_q\otimes\beta_q$ onto the $q$-th summand of Fact~\ref{fact:classical-factor}; its outcome maps are positive on all extensions (composites with a classical factor are direct sums), they sum to the identity, and their matrix entries are $0$ and $1$. The verifier can realise it with its gates by copying the register into fresh classical bits with reversible gates and reading the copies. We may therefore assume that the verifier reads out each outgoing message $Q_i$ before sending it and each incoming message $A_i$ on receipt; this refines its record and does not change the acceptance probability. Let $z$ denote the verifier's complete record, which includes the questions $q_1,\dots,q_r$ and the bit $\acc$, and let $a=(a_1,\dots,a_r)$ be the answers received. The verifier's circuit with the answers fixed on the wires $A_i$ is an adaptive circuit with classical input $a$, so $\Pr_V[z\,|\,a]$ is defined by \eqref{eq:record}. By Fact~\ref{fact:classical-io},
\begin{equation}\label{eq:prover-decomp}
\Pi_i(\delta_q\otimes\tau)=\sum_a\delta_a\otimes\Pi_i^{qa}(\tau),\qquad\Pi_i^{qa}\text{ positive on all extensions},\qquad\sum_au\circ\Pi_i^{qa}=u.
\end{equation}

\begin{lemma}[Product form]\label{lem:product}
In the interaction of $V_x$ with $P$, for every record $(z,a)$ the unnormalised joint state of the verifier's register and the prover's system at the end of round $i$, conditioned on the part of the record produced so far, is a product
\[
\sigma_i(z,a)\otimes\tau_i(q_1,a_1,\dots,q_i,a_i),\qquad\tau_i=\Pi_i^{q_ia_i}\circ\dots\circ\Pi_1^{q_1a_1}(\tau_0),
\]
where $\sigma_i$ depends only on the verifier's circuit, its record and the answers received. Consequently
\begin{equation}\label{eq:factor}
\Pr[z,a]=\Pr_V[z\,|\,a]\cdot\pi(a\,|\,q),\qquad\pi(a\,|\,q):=u\big(\tau_r(q_1,a_1,\dots,q_r,a_r)\big),
\end{equation}
where $q=(q_1,\dots,q_r)$ are the questions recorded in $z$.
\end{lemma}
\begin{proof}
Induction on $i$. Before round $1$ the joint state is $\omega_0\otimes\tau_0$. Suppose that the state after round $i-1$, conditioned on the record so far, is $\sigma\otimes\tau$. The verifier's stage-$i$ gates act as $M\otimes\id_P$ by (T1), giving $\sigma'\otimes\tau$. By Fact~\ref{fact:classical-factor} the verifier's register, which contains the classical wire $Q_i$, is in a state $\sum_q\sigma''_q\otimes\delta_q$, and reading $Q_i$ with result $q_i$ leaves $\sigma''_{q_i}\otimes\delta_{q_i}\otimes\tau$. The prover applies $\Pi_i$ to $Q_iP$, producing $\sigma''_{q_i}\otimes\sum_a\delta_a\otimes\Pi_i^{q_ia}(\tau)$ by \eqref{eq:prover-decomp}, and the verifier reads $A_i$ with result $a_i$, leaving $\sigma''_{q_i}\otimes\delta_{a_i}\otimes\Pi_i^{q_ia_i}(\tau)$, which is again a product. After the final stage the state is $\sigma_{\mathrm{final}}\otimes\tau_r$, where $u(\sigma_{\mathrm{final}})=\Pr_V[z\,|\,a]$ because the verifier's outcome maps have been applied with the answers fixed on the wires $A_i$, and the probability of the record is $u(\sigma_{\mathrm{final}})\,u(\tau_r)$.
\end{proof}

\begin{lemma}[Classical simulation of prover responses]\label{lem:classical-sim}
For every prover strategy $P$ in $\GG$ there are conditional distributions $\pi_i(a_i\,|\,q_1,a_1,\dots,q_{i-1},a_{i-1},q_i)$, $i=1,\dots,r$, such that $\Pr[z,a]=\Pr_V[z\,|\,a]\prod_i\pi_i(a_i\,|\,q_{\le i},a_{<i})$ for every record. Thus $P$ produces the same distribution of records, hence the same acceptance probability, as a classical randomised prover with memory that answers according to $\pi_1,\dots,\pi_r$. Conversely, every classical randomised prover is a strategy in $\GG$ by (T3). Hence
\[
\val_\GG(V_x)=\sup\{\Pr[V_x\leftrightarrow P\text{ accepts}]:P\text{ a classical randomised prover}\}.
\]
\end{lemma}
\begin{proof}
Put $\pi_i(a_i\,|\,q_{\le i},a_{<i})=u(\tau_i)/u(\tau_{i-1})$ when $u(\tau_{i-1})>0$, and an arbitrary distribution otherwise. The normalisation $\sum_au(\Pi_i^{qa}(\tau_{i-1}))=u(\tau_{i-1})$ in \eqref{eq:prover-decomp} shows that $\pi_i$ is a probability distribution in $a_i$ that depends only on the history. Telescoping, $\prod_i\pi_i=u(\tau_r)/u(\tau_0)=\pi(a\,|\,q)$ whenever all denominators are positive; if some $u(\tau_{i-1})=0$, then $\tau_{i-1}=0$ because $u$ is strictly positive on the cone, so $\tau_r=0$ and both sides vanish. Now apply \eqref{eq:factor}. For the converse, a classical randomised prover with memory keeps the transcript in a classical register and applies stochastic maps, which are allowed channels between classical systems by (T3).
\end{proof}

\begin{example}[A guessing game]\label{ex:guessing}
Let the verifier draw two uniformly random bits $y$ and $x$, keep $y$ in a classical register, send $x$ to the prover, and accept if and only if the prover's reply $a$ satisfies $a=x\oplus y$. A classical prover succeeds with probability $1/2$, since $y$ is uniform and independent of everything the prover sees. A prover holding half of a PR box, or any other system of any theory, cannot do better: by Lemma~\ref{lem:product} the joint state of the verifier's register and the prover's system is a product at every stage, conditioned on the recorded transcript, so the prover's reply is distributed as $\pi(a\,|\,x)$ for some conditional distribution, and $\Pr[a=x\oplus y]=\tfrac12$. Classical messages can create correlations between the two parties, but these are fully described by the classical response distributions of Lemma~\ref{lem:classical-sim}. The situation is different when there are two separated provers who share a nonclassical state, as in the CHSH game \cite{CHSH69}, where provers sharing a PR box win with probability $1$ \cite{PR94}; there the verifier can test correlations between systems it never touches.
\end{example}

Lemma~\ref{lem:classical-sim} shows that a nonclassical prover has no advantage over a classical prover when all messages are classical. The proof uses only positivity, normalisation (causality), the uniqueness of composites with a classical factor, and local tomography in the form ``the verifier's gates act as $M\otimes\id_P$''.

\begin{fact}[Deterministic provers suffice \cite{Fel86,Pap85,HO02}]\label{fact:deterministic}
For a verifier with classical messages define, for each complete transcript $(q,a)$, the weight $W(q,a)=\sum_{z\text{ records }q,\ \acc=1}\Pr_V[z\,|\,a]$, and define $F$ on partial transcripts by backward recursion: $F(q,a)=W(q,a)$; $F(h,q)=\max_aF(h,q,a)$ when the next move is an answer; $F(h)=\sum_qF(h,q)$ when the next move is a question. Then $\val_\GG(V_x)=F(\emptyset)$, and the value is attained by a deterministic prover $f=(f_1,\dots,f_r)$ with $a_i=f_i(q_1,\dots,q_i)$.
\end{fact}
\begin{proof}
By Lemma~\ref{lem:classical-sim} the acceptance probability of a classical prover $\pi$ is
\[
\mathrm{Acc}(\pi)=\sum_{(q,a)}W(q,a)\prod_i\pi_i(a_i\,|\,q_{\le i},a_{<i}).
\]
For fixed $i$ and history $h$, each transcript contains the factor $\pi_i(\cdot\,|\,h)$ at most once, so $\mathrm{Acc}$ is affine in the probability vector $\pi_i(\cdot\,|\,h)$ when the other conditionals are fixed. An affine function on a simplex is maximised at a vertex, so replacing the conditionals one at a time by point masses never decreases $\mathrm{Acc}$ and terminates with a deterministic prover. For a deterministic $f$, $\mathrm{Acc}(f)=\sum_{q_1}\cdots\sum_{q_r}W(q,f(q))$, and maximising over $f$, whose choice of $a_i$ may depend on $q_{\le i}$ separately for each history, gives the recursion defining $F(\emptyset)$.
\end{proof}

\subsection{The theorem}

We call a verifier with classical messages \emph{public-coin} if every question $Q_i$ is a string of fresh fair-coin outcomes, sent to the prover unchanged; the verifiers of the protocols of Fact~\ref{fact:IP} are of this form.

\begin{theorem}[Polynomial-space lower bound]\label{thm:lower}
For every theory $\GG$, $\PSPACE\subseteq\GIPcl(\GG)\subseteq\GIP(\GG)$. The same holds with the public-coin subclass of $\GIPcl(\GG)$ in place of $\GIPcl(\GG)$.
\end{theorem}
\begin{proof}
The second inclusion holds by definition. For the first, let $L\in\PSPACE$ and let $(V_S,P_S)$ be a public-coin interactive proof system for $L$ with perfect completeness and soundness $1/3$ (Fact~\ref{fact:IP}). By (T3) and (T4) the verifier $V_S$, a probabilistic polynomial-time machine, is a polynomial-time uniform family of adaptive circuits of $\GG$ using classical wires, fair coins, reversible gates and readouts, with all messages classical; call it $V$. If $x\in L_{\mathrm{yes}}$, the honest prover $P_S$ is a classical randomised prover, hence a strategy in $\GG$, and it makes $V_x$ accept with probability $1$. If $x\in L_{\mathrm{no}}$ and $P$ is any strategy in $\GG$, Lemma~\ref{lem:classical-sim} gives a classical randomised prover with the same distribution of records, which by soundness makes $V_S$ accept with probability at most $1/3$. Hence $\val_\GG(V_x)\le1/3$.
\end{proof}

\begin{theorem}[Classical messages]\label{thm:cl-universal}
For every theory $\GG$, $\GIPcl(\GG)=\PSPACE$. The public-coin subclass of $\GIPcl(\GG)$ also equals $\PSPACE$.
\end{theorem}
\begin{proof}
By Theorem~\ref{thm:lower} it remains to show $\GIPcl(\GG)\subseteq\PSPACE$. Let $L\in\GIPcl(\GG)$ with verifier $V$, and let $x$ have length $n$. Each $Q_i,A_i$ is a classical register of dimension at most $2^{p(n)}$, so questions and answers are strings of at most $p(n)$ bits, and $r\le p(n)$. We compute $F(\emptyset)$ of Fact~\ref{fact:deterministic} approximately in polynomial space and accept if and only if the result exceeds $1/2$.

\emph{Leaves.} $W(q,a)$ is a sum, over the polynomial-time decidable set of records that are consistent with $q$ and have $\acc=1$, of the probabilities $\Pr_V[z\,|\,a]$ of a polynomial-time uniform adaptive family with classical input $a$. By Fact~\ref{fact:pathsum}, for any polynomial $t$ a polynomial-space algorithm computes $\tilde W(q,a)$ with $|\tilde W(q,a)-W(q,a)|\le2^{-t(n)}$.

\emph{Recursion.} The tree of partial transcripts has depth $2r\le2p(n)$ and branching at most $2^{p(n)}$. We evaluate $F$ with $\tilde W$ in place of $W$ by depth-first search: at a question node we enumerate the $p(n)$-bit labels $q$ and accumulate a sum; at an answer node we enumerate the labels $a$ and keep a running maximum. Each of the $2r$ levels stores the current branch index and one accumulator, a rational number with $\poly(n)$ bits (at most $2^{\poly(n)}$ leaf values with $\poly(n)$ bits each are added per level). The value obtained is $\tilde F(\emptyset)=\max_f\sum_q\tilde W(q,f(q))$, and since each deterministic $f$ contributes at most $2^{rp(n)}$ terms,
\[
|\tilde F(\emptyset)-F(\emptyset)|\le2^{rp(n)}\,2^{-t(n)}\le\tfrac1{12}\qquad\text{for }t(n)=rp(n)+4.
\]
By Fact~\ref{fact:deterministic}, $F(\emptyset)=\val_\GG(V_x)$ is $\ge2/3$ or $\le1/3$, so $\tilde F(\emptyset)$ is $\ge7/12$ or $\le5/12$, and the comparison with $1/2$ decides correctly.

For the public-coin subclass, the upper bound is immediate and the lower bound is Theorem~\ref{thm:lower}.
\end{proof}

\begin{remark}[Where the theory enters]\label{rem:where-theory-enters}
The proof generalises Feldman's argument that $\IP\subseteq\PSPACE$ \cite{Fel86,HO02}, and the alternation of sums and maxima in Fact~\ref{fact:deterministic} is a ``game against nature'' in the sense of Papadimitriou \cite{Pap85}. {The theory $\GG$ determines the weights $W(q,a)$, which are computed using the path-sum representation, computable gates and the Kronecker structure of the verifier's own circuit.} The theorem therefore holds for boxworld, where the verifier may prepare and measure PR boxes internally; for $\GG_{\min}$ and $\GG_E$; for the theories $\GG_v$ and $\GG^{\mathrm Q}_v$ of Examples~\ref{ex:Gv} and~\ref{ex:GQv}, whatever the value of $v$; and for every other theory satisfying the standing assumptions.
\end{remark}

\begin{corollary}\label{cor:cl-iff}
For every theory $\GG$, $\GIP(\GG)=\PSPACE$ if and only if $\GIP(\GG)\subseteq\PSPACE$, i.e., if and only if every protocol with nonclassical messages can be simulated in polynomial space. Section~\ref{sec:beyond} shows that this can fail.
\end{corollary}

\begin{example}[Classical messages in $\GG_v$ and in boxworld]\label{ex:cl-in-Gv}
In the theory $\GG_v$ of Example~\ref{ex:Gv} the verifier may prepare copies of $e_0$ and measure $(f_y,u-f_y)$ internally, and the prover may send and receive only classical registers. Since $f_y(e_0)=0$, the verifier's own $W_v$-gates never produce a nontrivial outcome, and the real number $v$ plays no role: Theorem~\ref{thm:cl-universal} gives $\GIPcl(\GG_v)=\PSPACE$ for every $v$, even a noncomputable one. In boxworld the verifier may prepare PR boxes and measure them; its outcome probabilities are path sums with entries in $\{0,\pm1,\pm\tfrac12\}$, and again $\GIPcl(\text{boxworld})=\PSPACE$. Whether the class changes when the messages may be nonclassical is a different matter: for $\GG_v$ with a noncomputable $v$ it does (Theorem~\ref{thm:triangle}), for $\GG_v$ with an efficiently approximable $v$ it does not (Proposition~\ref{prop:presentation}), and for boxworld the question is open (Section~\ref{sec:discussion}).
\end{example}

\section{Quantum verifiers}\label{sec:qm}

We now let the verifier be an ordinary polynomial-time quantum circuit, with qubit registers as messages, and allow the prover to use the systems and channels of any theory $\GG$ containing quantum theory (Definition~\ref{def:containsQ}). The prover may hold systems that are not quantum systems, and its joint state with the verifier's qubits may lie outside the quantum state space, as in Example~\ref{ex:foil}. {We represent the prover's responses by \emph{quantum strategies} in the sense of Gutoski and Watrous} \cite{GW07}, also called \emph{quantum combs} \cite{CDP08,CDP09}. We use the name comb for the object and reserve ``strategy'' for the prover's strategy in the sense of Definition~\ref{def:prover}.

\subsection{Quantum combs}

\begin{definition}[Quantum combs]\label{def:comb}
Let $M_1,\dots,M_r$ and $M'_1,\dots,M'_r$ be qubit registers. An $r$-slot quantum comb with inputs $M_i$ and outputs $M'_i$ is a linear map $S:\bigotimes_iV_{M_i}\to\bigotimes_iV_{M'_i}$ for which there are linear maps $S^{(j)}:\bigotimes_{i\le j}V_{M_i}\to\bigotimes_{i\le j}V_{M'_i}$, $S^{(r)}=S$, $S^{(0)}=1$, such that (i) every $S^{(j)}$ is completely positive, and (ii) $(\id\otimes\Tr_{M'_j})\circ S^{(j)}=S^{(j-1)}\otimes\Tr_{M_j}$ for $1\le j\le r$. The comb \emph{induced} by a sequence of quantum channels with memory, $\Pi_i:M_i\otimes P_{i-1}\to M'_i\otimes P_i$, and an initial state $\tau_0$ on $P_0$ is
\[
S=(\id\otimes\Tr_{P_r})\circ(\Pi_r\otimes\id)\circ\dots\circ(\Pi_1\otimes\id)\circ(\tau_0\otimes\id),
\]
with the obvious reorderings of tensor factors, and $S^{(j)}$ is defined likewise with $r$ replaced by $j$.
\end{definition}

Condition (ii) is a causality condition: the output $M'_j$ of slot $j$, once discarded, leaves no trace of the input $M_j$, so the answers of the first $j-1$ rounds cannot depend on the questions of later rounds. In the language of Choi operators, (i) and (ii) are the semidefinite constraints that characterise the strategies of one player in a quantum game \cite{GW07}, and the value of a quantum interactive proof is the optimum of a semidefinite programme over them \cite{GW07,JJUW11}.

\begin{fact}[Realisation of quantum combs \cite{CDP09,GW07}]\label{fact:comb-realization}
Every $r$-slot quantum comb is induced by some sequence of quantum channels with finite-dimensional memory; indeed by isometries followed by a final partial trace. Conversely, every induced map is a quantum comb.
\end{fact}

\subsection{The theorem}

\begin{theorem}[Quantum verifiers]\label{thm:qm-universal}
Let $\GG$ be a theory containing quantum theory. For every quantum verifier $V$ and every input $x$,
\[
\val_\GG(V_x)=\val_{\mathrm{quantum}}(V_x),
\]
the value over quantum provers. Consequently $\GIPqm(\GG;k)=\QIP(k)$ for every $k$, and $\GIPqm(\GG)=\GIPqm(\GG;3)=\QIP=\PSPACE$.
\end{theorem}
\begin{proof}
Quantum strategies are strategies of $\GG$ by Definition~\ref{def:containsQ}(ii), so $\val_\GG\ge\val_{\mathrm{quantum}}$. For the converse, let $P=(\tau_0;P_0,\dots,P_r;\Pi_1,\dots,\Pi_r)$ be a strategy of $\GG$ against $V_x$, where $\Pi_i\in\Chan{M_iP_{i-1}}{M'_iP_i}$ with $M_i=Q^{\otimes m_i}$ the question register and $M'_i=Q^{\otimes m'_i}$ the answer register of round $i$. Define $S$ and $S^{(j)}$ from $(\tau_0,\Pi_i)$ by the formula of Definition~\ref{def:comb}; these are linear maps between the real spaces $\bigotimes_iV_{M_i}$ and $\bigotimes_iV_{M'_i}$, which by (T1) and Definition~\ref{def:containsQ}(i) are the spaces of Hermitian matrices on the corresponding qubit registers.

\emph{Claim 1: $S$ is a quantum comb.} Let $k\ge0$, $R=Q^{\otimes k}$, and let $\xi\in\Pos(M_1\cdots M_rR)$. By Definition~\ref{def:containsQ}(i), $\xi\in\Vp{M_1\cdots M_rR}$. Products of states lie in composites (Definition~\ref{def:composite}), so $\tau_0\otimes\xi\in\Vp{P_0M_1\cdots M_rR}$. Each $\Pi_i\otimes\id$ is positive on all extensions of $\GG$, and the systems $M_{i+1}\cdots M_rR$ and $M'_1\cdots M'_{i-1}$ on which it acts as the identity are systems of $\GG$; the final discard of $P_r$ is likewise an allowed channel. Hence $(S\otimes\id_R)(\xi)\in\Vp{M'_1\cdots M'_rR}=\Pos(M'_1\cdots M'_rR)$. Thus $S\otimes\id_{M_{2^k}(\C)}$ is positive for every $k$, so $S$ is completely positive by Choi's theorem \cite{Cho75}; the same argument applies to each $S^{(j)}$. Condition (ii) of Definition~\ref{def:comb} follows from causality: since $\Pi_j$ is a channel, $(\Tr_{M'_j}\otimes u_{P_j})\circ\Pi_j=\Tr_{M_j}\otimes u_{P_{j-1}}$, so discarding $M'_j$ and $P_j$ after $\Pi_j$ is the same as discarding $M_j$ and $P_{j-1}$ before it.

\emph{Claim 2: the acceptance probability depends only on $S$.} By (T1), all state spaces in the interaction are tensor products of the spaces of the systems involved, and every operation of the verifier and of the prover is a linear map between such tensor products. Expanding all maps in product bases, the acceptance probability is the full contraction of a tensor network. The prover's tensors ($\tau_0$, $\Pi_1,\dots,\Pi_r$, $u_{P_r}$) are joined to one another by the indices of $P_0,\dots,P_r$ and to the verifier's tensors only by the indices of $M_1,\dots,M_r,M'_1,\dots,M'_r$; the verifier's tensors carry no $P$-index. A tensor contraction may be carried out in any order, and contracting the $P$-indices first replaces the prover's tensors by the single tensor $S$. Hence the acceptance probability is a fixed linear functional, determined by $V_x$, evaluated on $S$.

\emph{Claim 3: a quantum prover induces $S$.} By Claim~1 and Fact~\ref{fact:comb-realization} there is a quantum strategy $\tilde P=(\rho_0;\tilde P_0,\dots,\tilde P_r;\tilde\Pi_1,\dots,\tilde\Pi_r)$ with finite-dimensional quantum memories that induces the same comb $S$. By Claim~2, $\Pr[V_x\leftrightarrow\tilde P\text{ accepts}]=\Pr[V_x\leftrightarrow P\text{ accepts}]$.

Taking the supremum over $P$ gives $\val_\GG(V_x)\le\val_{\mathrm{quantum}}(V_x)$. For the class statements, note that the quantum verifiers for $\GG$ are exactly the $\QIP$ verifiers, with the same count of messages, and that for each of them the two values coincide; so the same verifiers witness membership in $\GIPqm(\GG;k)$ and in $\QIP(k)$, and Fact~\ref{fact:QIP} completes the proof.
\end{proof}

\begin{remark}\label{rem:qm}
(i) The theorem needs neither a bound on the prover's memory, nor a purification principle in $\GG$, nor any information about the composites of quantum systems with the other systems of $\GG$. It uses only {local tomography (T1), that} allowed channels are positive on every extension and preserve the unit effect (T2), that products of states are states, and that composites of qubits are the ordinary quantum composites. The realisation theorem then supplies a quantum prover with the same acceptance probability. (ii) With classical messages the analogous statement is Lemma~\ref{lem:classical-sim}; there the induced object is a classical conditional distribution, and its realisation by a classical prover is immediate. (iii) The last hypothesis, that composites of qubits are the quantum composites, cannot be dropped: Example~\ref{ex:unot} exhibits a verifier that uses only single-qubit operations and whose value in $\GG_{\min}$ exceeds its quantum value.
\end{remark}

{For bipartite correlations, Barnum et al. show that local quantum measurements and no-signalling imply quantum correlations \cite{BBBEW10}. Theorem~\ref{thm:qm-universal} concerns sequential strategies with quantum messages.}

\subsection{Examples}

\begin{example}[One round]\label{ex:one-round}
In a single question--answer round the comb is the channel $S=(\id\otimes u_{P_1})\circ\Pi_1\circ(\id\otimes\tau_0)$ from the question register to the answer register, obtained by preparing the prover's memory, applying its channel and discarding the memory. Claim~1 says that positivity of $\Pi_1$ on every extension by a qubit register makes $S$ completely positive, so a prover of $\GG$ acts on the message exactly as some quantum channel. With several rounds one must also ensure that no answer depends on later questions; condition (ii) of Definition~\ref{def:comb} expresses exactly this.
\end{example}

\begin{example}[A prover of $\GG_E$]\label{ex:foil-steering}
Let the prover of the theory $\GG_E$ (Example~\ref{ex:foil}) prepare the state $\tfrac12\mathrm{SWAP}$ on $Q\otimes E$, send $Q$ to the verifier and keep $E$. Any measurement the prover later performs on $E$, with effects $(g_a)_a$ satisfying $0\le g_a\le\one$ and $\sum_ag_a=\one$, steers the verifier's qubit to the unnormalised states
\[
(\id_Q\otimes g_a)(\tfrac12\mathrm{SWAP})=\tfrac12\Tr_E\big[(\one\otimes g_a)\,\mathrm{SWAP}\big]=\tfrac12\,g_a,
\]
using $\mathrm{SWAP}=\sum_{i,j}\ket{i}\bra{j}\otimes\ket{j}\bra{i}$. A quantum prover that holds half of the maximally entangled state $\Phi^+$ and measures its half with the transposed effects $(g_a^{\mathsf T})_a$ steers the verifier's qubit to $(\id\otimes g_a^{\mathsf T})(\Phi^+)=\tfrac12g_a$, exactly the same states. In the language of Theorem~\ref{thm:qm-universal}, the two strategies induce the same comb, so the state outside the quantum cone gives the prover of $\GG_E$ no advantage.
\end{example}

\begin{example}[A spin-flip test]\label{ex:unot}
Let $\sigma_x,\sigma_y,\sigma_z$ be the Pauli matrices. The verifier tosses three fair coins $b_1,b_2,b_3$: it chooses the axis $a=z$ if $b_1=0$, and $a=x$ or $a=y$ according to $b_2$ if $b_1=1$, and the sign $s=(-1)^{b_3}$. It prepares a qubit in the eigenstate $\ket\psi$ of $\sigma_a$ with eigenvalue $s$, sends it to the prover, receives a qubit back, measures $\sigma_a$ on it, and accepts if and only if the outcome is $-s$. In words, the verifier asks the prover to flip the spin: to map $\ket\psi$ to its orthogonal complement. The six states $\ket0,\ket1,\ket\pm,\ket{\pm i}$ are prepared and measured with the gates $H$ and $T$ (using $S=T^2$ and $X=HT^4H$), $\ket0$-preparation and computational-basis measurement on a single qubit register, so this is a verifier for $\QC$ and also for $\GG_{\min}$ of Example~\ref{ex:Gmin}. Quantum theory has no universal spin flip \cite{BHW99,GP99}, and we now compute the value in both theories.

\emph{In $\QC$, and in every theory containing quantum theory, the value is $3/4$.} By Theorem~\ref{thm:qm-universal} we may take the prover to be quantum, and by Example~\ref{ex:one-round} its action on the message is a qubit channel $\Phi$. Write $\Phi$ in Bloch form: the state with Bloch vector $n$ is mapped to the state with Bloch vector $Mn+t$, where $M$ is a real $3\times3$ matrix. On the input $\ket\psi$ with Bloch vector $se_a$, the measurement of $\sigma_a$ gives the outcome $-s$ with probability $\tfrac12(1-M_{aa}-s\,e_a\cdot t)$, and averaging over $s=\pm1$ removes $t$. With the weights $w_z=\tfrac12$ and $w_x=w_y=\tfrac14$ of the three axes, the acceptance probability is
\[
\tfrac12-\tfrac12\big(\tfrac14M_{xx}+\tfrac14M_{yy}+\tfrac12M_{zz}\big).
\]
The diagonal $(M_{xx},M_{yy},M_{zz})$ of a qubit channel lies in the tetrahedron
\[
T=\conv\{(1,1,1),\ (1,-1,-1),\ (-1,1,-1),\ (-1,-1,1)\}.
\]
To see this, let $J=(\Phi\otimes\id)(\ketbra{\Phi^+})$ be the Choi state of $\Phi$, which is positive semidefinite with unit trace, let $\ket{B_k}=(\sigma_k\otimes\one)\ket{\Phi^+}$, $k\in\{0,x,y,z\}$ with $\sigma_0=\one$, be the Bell basis, and let $p_k=\bra{B_k}J\ket{B_k}\ge0$, so that $\sum_kp_k=1$. The identity $\Tr[(A\otimes B)J]=\tfrac12\Tr[A\,\Phi(B^{\mathsf T})]$ gives $M_{aa}=\tfrac12\Tr[\sigma_a\Phi(\sigma_a)]=\Tr[(\sigma_a\otimes\sigma_a^{\mathsf T})J]$, and $(\sigma_a\otimes\sigma_a^{\mathsf T})\ket{B_k}=(\sigma_a\sigma_k\sigma_a\otimes\one)\ket{\Phi^+}=\pm\ket{B_k}$ with the sign $+$ exactly when $k\in\{0,a\}$. Hence $M_{aa}=\sum_k\pm p_k$, and
\[
(M_{xx},M_{yy},M_{zz})=p_0(1,1,1)+p_x(1,-1,-1)+p_y(-1,1,-1)+p_z(-1,-1,1)\in T.
\]
The linear function $\tfrac14\lambda_x+\tfrac14\lambda_y+\tfrac12\lambda_z$ takes the values $1,-\tfrac12,-\tfrac12,0$ at the four vertices, so it is at least $-\tfrac12$ on $T$, and the acceptance probability is at most $\tfrac12+\tfrac14=\tfrac34$. The bound is attained by $\Phi(\rho)=\sigma_x\rho\sigma_x$, which flips $\ket0,\ket1,\ket{+i},\ket{-i}$ and fixes $\ket\pm$: its acceptance probability is $\tfrac12+\tfrac14=\tfrac34$.

\emph{In $\GG_{\min}$ the value is $1$.} The map $R(\rho)=\Tr(\rho)\one-\rho$ is a channel of $\GG_{\min}$ (Example~\ref{ex:Gmin}), and $R(\ketbra\psi)=\one-\ketbra\psi$ is the orthogonal complement of every pure qubit state $\ket\psi$. A prover applying $R$ to the message is accepted with probability $1$.

The two theories have the same elementary systems and the same single-qubit gates, and the verifier is the same circuit. What differs is the composition rule for two qubits, which decides whether $R$ is positive on all extensions. Theorem~\ref{thm:qm-universal} applies only when the composites of qubits are the quantum ones; then not even a prover of $\GG_E$ holding the state $\tfrac12\mathrm{SWAP}$ of Example~\ref{ex:foil-steering} can exceed $3/4$.
\end{example}

\subsection{Consequences}

\begin{corollary}[Instances]\label{cor:instances}
The general theorems specialise as follows.
\begin{enumerate}[label=(\roman*)]
\item For $\GG=$ classical theory, Theorem~\ref{thm:cl-universal} reads $\IP=\PSPACE$; for $\GG=\QC$, Theorem~\ref{thm:qm-universal} reads $\QIP=\PSPACE$ and $\QIP(3)=\PSPACE$.
\item For every theory $\GG$ containing quantum theory, $\GIPcl(\GG)=\GIPqm(\GG)=\PSPACE$, and $\GIP(\GG)\supseteq\PSPACE$, with equality if and only if all protocols outside the classical-message subclass and the quantum-verifier subclass can be simulated in polynomial space.
\item $\GIPqm(\GG_E)=\PSPACE$ for the theory of Example~\ref{ex:foil}.
\end{enumerate}
\end{corollary}
\begin{proof}
(i) is the definition of the classes. (ii) combines Theorems~\ref{thm:lower}, \ref{thm:cl-universal} and~\ref{thm:qm-universal}. For (iii), apply Theorem~\ref{thm:qm-universal}, since $\GG_E$ contains quantum theory.
\end{proof}

The lower bound in (i) is Shamir's theorem, imported as a black box in Theorem~\ref{thm:lower}, and the upper bound $\QIP\subseteq\PSPACE$ is the theorem of Jain, Ji, Upadhyay and Watrous, imported in Theorem~\ref{thm:qm-universal}. What the general theorems add is the reduction of arbitrary provers to classical or quantum ones, and, for classical messages, a theory-independent polynomial-space upper bound.

{In quantum theory with classical systems, allowing classical registers as messages, in addition to qubit registers, does not change the class. We use this in Section~\ref{sec:beyond}.}

\begin{lemma}[Classical registers as messages in quantum theory]\label{lem:QC}
Let $\QC$ be quantum theory with classical systems, possibly with finitely many classical systems $\Delta_d$ of fixed dimensions adjoined as elementary systems, each with its standard coordinates and with the gates preparation of $\delta_1$, readout of the label and discard. Then $\GIP(\QC;k)=\QIP(k)$ for every $k$, and $\GIP(\QC)=\QIP=\PSPACE$. {The fixed-message equality holds with any fixed thresholds $0<s<c<1$ imposed on both sides; for polynomially many messages these thresholds give the same class $\PSPACE$.}
\end{lemma}
\begin{proof}
$\QIP(k)\subseteq\GIP(\QC;k)$ because a $\QIP$ verifier is a $\QC$-verifier and, by Theorem~\ref{thm:qm-universal}, its value over $\QC$-provers equals its quantum value. For the converse, let $V$ be a verifier for $\QC$. Its message registers are composites of qubits and classical systems, i.e., classical--quantum registers with a classical label $c$ and a quantum part. Define a quantum verifier $V'$ by replacing each classical label register in a message by a qubit register of $\lceil\log_2d\rceil$ qubits per elementary classical factor, which $V'$ prepares in the computational-basis state $\ket c$ when $V$ would send the label $c$, and which $V'$ measures in the computational basis immediately upon receipt, treating the outcome as the label received (an outcome that does not encode a label is treated as the label $1$); the rest of the circuit is unchanged, the verifier's internal classical registers being part of the classical control of the quantum circuit. Then $V'$ has the same number of messages as $V$, and we claim $\val_{\mathrm{quantum}}(V'_x)=\val_{\QC}(V_x)$.

Given a $\QC$-prover $P$ for $V$, with channels $\Pi_i$, define a quantum prover $P'$ for $V'$ whose $i$-th channel first measures the label qubits of the incoming question in the computational basis (which does not disturb the computational-basis state $\ket c$ that $V'$ sends), then applies $\Pi_i$, and finally encodes the classical label of the answer as a computational-basis state; this is a CPTP map. (The classical parts of the prover's private memory are encoded in computational-basis states in the same way, so that $P'$ has quantum memory.) By induction on the rounds, the joint states in $V'\leftrightarrow P'$ are the encodings of those in $V\leftrightarrow P$, so the acceptance probabilities coincide. Conversely, given a quantum prover $P'$ for $V'$, define a $\QC$-prover $P$ for $V$ whose $i$-th channel reads the classical label of the question, prepares $\ket c$ on label qubits, applies the channel of $P'$, and measures the label qubits of the answer in the computational basis, recording the outcome in the classical answer register with the same convention for outcomes that do not encode a label. Since $V'$ measures the label qubits immediately anyway, the interactions $V\leftrightarrow P$ and $V'\leftrightarrow P'$ again have the same joint states after every round and the same acceptance probability. Hence $\val_{\QC}(V_x)=\val_{\mathrm{quantum}}(V'_x)$, and $L\in\GIP(\QC;k)$ implies $L\in\QIP(k)$ with the same thresholds. For polynomially many messages, Proposition~\ref{prop:error} applied to $\QC$ reduces arbitrary constant thresholds to $(2/3,1/3)$, and Fact~\ref{fact:QIP} gives $\PSPACE$.
\end{proof}

\section{Arbitrary messages: beyond polynomial space}\label{sec:beyond}

The standing assumptions constrain the verifier's gates but not the state spaces of the systems it measures. We now use this freedom to construct theories in which the geometry of a message state space encodes a real number that is hard, or impossible, to compute. A prover can send states that allow the verifier to estimate this number, even when the theory is isomorphic to quantum theory with classical systems and contains both quantum theory and classical theory as subtheories. {We then show that the phenomenon disappears when the transported verifier gates can be simulated efficiently by the standard gates.}

\subsection{The triangular theory and the problem \texorpdfstring{$L_v$}{L\_v}}\label{sec:triangle}

Recall the triangular system $W_v$ of Example~\ref{ex:triangle} and the theory $\GG_v$ of Example~\ref{ex:Gv}. For $v\in(0,1)$ define the promise problem $L_v=(L_{v,\mathrm{yes}},L_{v,\mathrm{no}})$ by
\[
L_{v,\mathrm{yes}}=\{t\in\{0,1\}^{+}:v\ge\bar t+|t|^{-2}\},\qquad L_{v,\mathrm{no}}=\{t\in\{0,1\}^{+}:v\le\bar t-|t|^{-2}\},
\]
where $\bar t=\sum_{j=1}^{|t|}t_j2^{-j}\in[0,1)$ is the binary fraction with digits $t$. The promise gap shrinks polynomially with the input length, so an algorithm for $L_v$ must in effect approximate $v$ to inverse-polynomial precision. A real number $v$ is \emph{computable} if some algorithm, given $k$, outputs a rational number within $2^{-k}$ of $v$ \cite{Tur37,Wei00}; a sequence of rationals $(w_m)$ with $|v-w_m|\le\eps(m)$ for a computable $\eps(m)\to0$ then shows that $v$ is computable whenever $m\mapsto w_m$ is computable.

\begin{theorem}[One prover message and a noncomputable coordinate]\label{thm:triangle}
Let $v\in(0,1)$ and let $\GG_v$ be the theory of Example~\ref{ex:Gv}. Then:
\begin{enumerate}[label=(\roman*)]
\item $\GG_v$ satisfies (T1)--(T4), and $\GIPcl(\GG_v)=\PSPACE$.
\item $L_v\in\GMA(\GG_v)$: there is a one-message protocol in which the prover sends $|t|^4$ copies of $W_v$ and the verifier measures each copy with $(f_y,u-f_y)$.
\item If $v$ is not computable, then $L_v$ is undecidable. Hence $\GMA(\GG_v)$ contains an undecidable promise problem, $\GIP(\GG_v)$ strictly contains $\PSPACE$, and $\GIPcl(\GG_v)\subsetneq\GIP(\GG_v)$.
\end{enumerate}
\end{theorem}
\begin{proof}
(i) was verified in Example~\ref{ex:Gv}; the class statement is Theorem~\ref{thm:cl-universal}.

(ii) On input $t$ with $m=|t|$, let $k=m^4$ and $\eps=m^{-2}$. The verifier expects the message system $W_v^{\otimes k}$, measures wire $i$ with $(f_y,u-f_y)$ to obtain $s_i\in\{0,1\}$, and accepts if and only if $\sum_is_i\ge k\bar t$. This is a polynomial-time uniform circuit of $\GG_v$: it uses $k$ measurement gates, classical counting, and one comparison with the rational number $k\bar t$.

\emph{Completeness.} If $v\ge\bar t+\eps$, the honest prover sends $e_2^{\otimes k}$, a state of $W_v^{\otimes k}$. The product effects $f_y^{(i)}$ evaluated on the product state give independent Bernoulli variables with $\Pr[s_i=1]=f_y(e_2)=v$, so by Hoeffding's inequality \cite{Hoe63}
\[
\Pr\Big[\sum_is_i<k\bar t\Big]\le\Pr\Big[\sum_is_i-kv<-k\eps\Big]\le e^{-2k\eps^2}=e^{-2}<\tfrac13.
\]

\emph{Soundness.} Suppose $v\le\bar t-\eps$, and let $\omega\in\Om{W_v^{\otimes k}}$ be the state of the message. The prover's private system, on which the verifier does not act, may be discarded by Fact~\ref{fact:causal}, and by Fact~\ref{fact:marginal} the resulting marginal is a state of $W_v^{\otimes k}$. For $s\in\{0,1\}^i$ write $f^{s}=f^{s_1}\otimes\dots\otimes f^{s_i}\otimes u\otimes\dots\otimes u$ with $f^1=f_y$ and $f^0=u-f_y$; then $\Pr[s_1\cdots s_i=s]=f^s(\omega)$, since each measurement acts on its own wire. By Fact~\ref{fact:marginal}, conditioned on outcomes $s_{<i}$ of positive probability, wire $i$ is in a state $\omega_i\in\Om{W_v}$, and
\[
\Pr[s_i=1\,|\,s_{<i}]=f_y(\omega_i)\le\max_{\Om{W_v}}f_y=v,
\]
because the maximum of the linear functional $f_y$ over the triangle is attained at a vertex (Fact~\ref{fact:convex}(iv)) and $f_y$ takes the values $0,0,v$ on the vertices. By Fact~\ref{fact:domination} and Hoeffding's inequality,
\[
\Pr\Big[\sum_is_i\ge k\bar t\Big]\le\Pr\big[\mathrm{Bin}(k,v)\ge k\bar t\big]\le\Pr\big[\mathrm{Bin}(k,v)\ge kv+k\eps\big]\le e^{-2k\eps^2}=e^{-2}<\tfrac13.
\]

(iii) Suppose that a total algorithm $D$ accepts every $t\in L_{v,\mathrm{yes}}$ and rejects every $t\in L_{v,\mathrm{no}}$. For each $m$ let $w_m$ be the largest $\bar t$ over the strings $t\in\{0,1\}^m$ accepted by $D$, or $0$ if none is accepted. Lemma~\ref{lem:approx} below shows that $|v-w_m|<m^{-2}+2^{-m}$; since $m\mapsto w_m$ is computable, $v$ is computable. Contrapositively, if $v$ is not computable, then $L_v$ is undecidable. It lies in $\GMA(\GG_v)\subseteq\GIP(\GG_v)$ by (ii), and $\PSPACE$ contains only decidable problems. The last claim follows from (i).
\end{proof}

\begin{lemma}[Approximating $v$ from a decider]\label{lem:approx}
Let $v\in(0,1)$ and let $D$ be a total algorithm that accepts every string of $L_{v,\mathrm{yes}}$ and rejects every string of $L_{v,\mathrm{no}}$. Let $w_m=\max\{\bar t:t\in\{0,1\}^m,\ D\text{ accepts }t\}$, with $w_m=0$ if $D$ accepts no string of length $m$. Then $|v-w_m|<m^{-2}+2^{-m}$. If moreover $D$ runs in space $s(m)\ge m$ on inputs of length $m$, then $w_m$ is computable from $m$ in space $O(s(m))$.
\end{lemma}
\begin{proof}
Every accepted $t$ satisfies $\bar t<v+m^{-2}$: otherwise $v\le\bar t-m^{-2}$, so $t\in L_{v,\mathrm{no}}$ and $D$ rejects it. Hence $w_m<v+m^{-2}$. If $v\ge m^{-2}$, let $t'$ be the largest $m$-bit string with $\bar t'\le v-m^{-2}$; then $\bar t'>v-m^{-2}-2^{-m}$ and $t'\in L_{v,\mathrm{yes}}$, so $D$ accepts $t'$ and $w_m\ge\bar t'>v-m^{-2}-2^{-m}$. If $v<m^{-2}$, then $w_m\ge0>v-m^{-2}$. Computing $w_m$ requires enumerating the $2^m$ strings of length $m$ and running $D$ on each while keeping the current maximum, which uses space $O(s(m)+m)$.
\end{proof}

{\paragraph{Relation to earlier one-message bounds.}
Lee and Hoban state $\GMA\subseteq\cls{A_0PP}\subseteq\PP$ for causal, locally tomographic theories in which all mathematically allowed states are physical \cite[Theorem~5.0.4]{LH16}. Their classes are defined for languages, whereas $L_v$ is a promise problem, so Theorem~\ref{thm:triangle} does not directly contradict that statement. The singular-value argument in their proof does not establish the corresponding bound here. In general, the acceptance probability $uM\rho$ need not be bounded by the largest singular value of $M$: for the classical outcome map
\[
M=\begin{pmatrix}1/2&0\\1/2&0\end{pmatrix},\qquad u=(1,1),\qquad \rho=(1,0)^{\mathsf T},
\]
we have $uM\rho=1$ but $\sigma_{\max}(M)=1/\sqrt2$. This is a valid outcome map: its entries are nonnegative and its column sums are at most one. Moreover, the change of coordinates used in that proof need not preserve efficient computability of gate entries. The examples here show that such an upper bound for promise problems requires an additional hypothesis beyond (T1)--(T4).}

\begin{example}[The protocol at $v=1/2$]\label{ex:half}
For $v=\tfrac12$ the state space of $W_{1/2}$ is the triangle with vertices $(0,0,1)$, $(1,0,1)$ and $(1,\tfrac12,1)$. A mixture of the vertices with probabilities $p_0,p_1,p_2$ gives $f_y=p_2/2$, so the verifier's measurement can be simulated in the canonical coordinates of $\Delta_3$ by reading the vertex label and tossing a fair coin when the label is $2$. The protocol of Theorem~\ref{thm:triangle}(ii) then estimates the bias $\tfrac12$ of this coin, and $L_{1/2}$ is trivially in $\P$. For a general $v$ the same simulation needs a coin of bias $v$, which is where the difficulty resides: the verifier's measurement gate has the fixed integer coordinates $(0,1,0)$ in the given presentation, but its maximum over the state space is $v$.
\end{example}

\begin{example}[A noncomputable coordinate]\label{ex:halting}
Let $H\subseteq\N$ be the set of indices of Turing machines that halt on the empty input, and let $v_H=\sum_{i\in H}2^{-i-2}$. Then $v_H\in(0,\tfrac12]$, its binary digits are the indicator of $H$, and $v_H$ is not computable: an algorithm computing $v_H$ to precision $2^{-k}$ for every $k$ would decide $H${. Since $H$ is neither finite nor cofinite, $v_H$ is not dyadic, so its binary digits can be recovered successively from sufficiently precise approximations \cite{Tur37,Wei00}.} Theorem~\ref{thm:triangle} then says that the verifier of $\GG_{v_H}$, whose gates have integer and dyadic entries, decides the undecidable problem $L_{v_H}$ with the help of a single message. {The honest prover sends copies of the fixed pure state $e_2$, independently of the contents of the input.} The verifier reads off the answer from the statistics of a fixed measurement on them.
\end{example}

\subsection{A theory containing quantum theory and classical theory}\label{sec:Qtriangle}

The theory $\GG_v$ contains no quantum systems. The next theorem adjoins $W_v$ to quantum theory and obtains a theory that contains both quantum theory and classical theory as subtheories, is isomorphic to $\QC$ with $\Delta_3$ adjoined, and has $\GIPcl=\GIPqm=\PSPACE$ by Theorems~\ref{thm:cl-universal} and~\ref{thm:qm-universal}, and yet has $\GIP\supsetneq\PSPACE$ for noncomputable $v$, and also for the computable $v_T$ of Theorem~\ref{thm:computable} below.

\begin{theorem}[Quantum theory with a triangular system adjoined]\label{thm:Qtriangle}
Let $v\in(0,1)$ and let $\GG^{\mathrm Q}_v$ be the theory of Example~\ref{ex:GQv}.
\begin{enumerate}[label=(\roman*)]
\item $\GG^{\mathrm Q}_v$ satisfies (T1)--(T4), contains quantum theory in the sense of Definition~\ref{def:containsQ} and classical theory in the sense of (T3), and is isomorphic as a theory to $\QC$ with $\Delta_3$ adjoined as an elementary system.
\item $\GIPcl(\GG^{\mathrm Q}_v)=\GIPqm(\GG^{\mathrm Q}_v)=\PSPACE$; the quantum subtheory has $\GIP=\QIP=\PSPACE$, and the classical subtheory has $\GIP=\IP=\PSPACE$.
\item $L_v\in\GMA(\GG^{\mathrm Q}_v)$.
\item If $v$ is not computable, then $\GIP(\GG^{\mathrm Q}_v)$ contains an undecidable promise problem, so $\GIP(\GG^{\mathrm Q}_v)\supsetneq\PSPACE$.
\end{enumerate}
\end{theorem}
\begin{proof}
(i) The standing assumptions and the isomorphism were verified in Example~\ref{ex:GQv}. Definition~\ref{def:containsQ}(i) holds because the transport maps act as the identity on qubit factors, so composites of qubits are the ordinary ones; (ii) holds because all CPTP maps of $\QC$ are transported; (iii) holds by the choice of the gate set.

(ii) Theorems~\ref{thm:cl-universal} and~\ref{thm:qm-universal} apply. The subtheory generated by qubits and classical bits with the quantum and classical gates is $\QC$, whose class is $\QIP$ by Lemma~\ref{lem:QC}, and the subtheory generated by classical bits with the classical gates is classical theory, whose class is $\IP$.

(iii) Run the protocol of Theorem~\ref{thm:triangle}(ii). Completeness is unchanged. For soundness, the proof of Theorem~\ref{thm:triangle} used only that the conditional state of each wire is a state of $W_v$ (Fact~\ref{fact:marginal}, which applies because every composite of $\GG^{\mathrm Q}_v$ lies inside the maximal tensor product) and that $\max_{\Om{W_v}}f_y=v$; both hold here. Equivalently, in the canonical coordinates the message is a classical register over $\{0,1,2\}^k$, possibly correlated with the prover's quantum memory, and each measurement outputs $1$ with probability $v$ when the label is $2$ and with probability $0$ otherwise, so $\Pr[s_i=1\,|\,s_{<i}]\le v$.

(iv) By (iii) and Theorem~\ref{thm:triangle}(iii).
\end{proof}

Parts (ii) and (iv) show that classical-message protocols, quantum-verifier protocols, and protocols restricted to the quantum or classical subtheories all remain in $\PSPACE$; the separation uses protocols in which the verifier measures the triangular systems sent by the prover.

\subsection{Computable presentations do not suffice}\label{sec:computable}

We next ask whether computability of the theory in the given coordinates suffices for $\GIP=\PSPACE$. It does not: a computable coordinate can still separate the class from $\PSPACE$. We use the unary space hierarchy theorem \cite{HLS65,AB09}, whose proof we recall in Appendix~\ref{app:hierarchy}.

\begin{fact}[Unary space hierarchy]\label{fact:hierarchy}
There is a set $T\subseteq\N$ such that the unary language $\{1^n:n\in T\}$ lies in $\DSPACE(2^{n^2})$ but not in $\DSPACE(2^{cn})$ for any constant $c$.
\end{fact}

\begin{theorem}[A computable coordinate outside polynomial space]\label{thm:computable}
Let $T$ be as in Fact~\ref{fact:hierarchy}, modified to contain $0$, and let $v_T=\sum_{i\in T}4^{-i-1}\in[\tfrac14,\tfrac13]$. Then:
\begin{enumerate}[label=(\roman*)]
\item $v_T$ is computable; {the state cones of $\GG_{v_T}$ and of $\GG^{\mathrm Q}_{v_T}$ have finite descriptions by linear or positive-semidefinite constraints with computable coefficients, and their unit effects, verifier gates and isomorphisms with classical theory and with $\QC$ have computable entries;}
\item $L_{v_T}\in\GMA(\GG_{v_T})\cap\GMA(\GG^{\mathrm Q}_{v_T})$;
\item $L_{v_T}\notin\PSPACE$;
\item $L_{v_T}$ is decidable, in space $2^{O(\log^2m)}$ on inputs of length $m$.
\end{enumerate}
Consequently $\GIP(\GG_{v_T})\supsetneq\PSPACE$ and $\GIP(\GG^{\mathrm Q}_{v_T})\supsetneq\PSPACE$ unconditionally, although both theories have computable presentations, $\GG^{\mathrm Q}_{v_T}$ contains quantum theory and classical theory, $\GIPcl=\PSPACE$ in both, and $\GIPqm(\GG^{\mathrm Q}_{v_T})=\PSPACE$.
\end{theorem}
\begin{proof}
(i) The base-$4$ digits of $v_T$ are the indicators $[i\in T]$, computable in space $2^{i^2}$; {the coefficients in these descriptions and the entries of the remaining maps are built from $v_T$ and the fixed computable coefficients of $\QC$ by finitely many rational operations.}

(ii) Theorems~\ref{thm:triangle}(ii) and~\ref{thm:Qtriangle}(iii).

(iii) Suppose that $L_{v_T}$ is decided by a total algorithm $D$ in space $m^d$. We decide $T$ on input $1^i$ in space $2^{O(i)}$, contradicting Fact~\ref{fact:hierarchy}. For $j=0,1,\dots,i$ in turn: let $m_j=2^{j+3}$; compute $w_{m_j}$ by Lemma~\ref{lem:approx}, in space $O(m_j^d)=2^{O(j)}$; compute $T_{j-1}=\sum_{l<j,\,l\in T}4^{-l-1}$ from the bits $[l\in T]$, $l<j$, found in the earlier iterations and kept in $j$ bits of storage; and declare $j\in T$ if and only if $w_{m_j}-T_{j-1}\ge\tfrac23\cdot4^{-j-1}$. The test is correct: by Lemma~\ref{lem:approx}, $|v_T-w_{m_j}|<m_j^{-2}+2^{-m_j}$, and $m_j^{-2}=\tfrac1{16}4^{-j-1}$ while $2^{-m_j}\le\tfrac18\cdot4^{-j-1}$ for all $j\ge0$, so $|v_T-w_{m_j}|<\tfrac14\cdot4^{-j-1}$. If $j\in T$, then $v_T-T_{j-1}\ge4^{-j-1}$, so $w_{m_j}-T_{j-1}>\tfrac34\cdot4^{-j-1}$; if $j\notin T$, then $v_T-T_{j-1}\le\sum_{l>j}4^{-l-1}=\tfrac13\cdot4^{-j-1}$, so $w_{m_j}-T_{j-1}<\tfrac7{12}\cdot4^{-j-1}$. The space used is $2^{O(i)}$ in total, since each iteration reuses the same work space and the numbers involved have $2^{O(i)}$ bits.

(iv) On input $t$ of length $m$, compute the digits $[j\in T]$ for $j\le J=\lceil\log_4(12m^2)\rceil$, each in space $2^{j^2}\le2^{O(\log^2m)}$, and let $\tilde v=\sum_{j\le J,\,j\in T}4^{-j-1}$; then $0\le v_T-\tilde v\le\tfrac13\cdot4^{-J-1}<\tfrac14m^{-2}$. Accept if and only if $\tilde v\ge\bar t+\tfrac12m^{-2}$. If $t\in L_{v_T,\mathrm{yes}}$, then $\tilde v\ge\bar t+\tfrac34m^{-2}$, and if $t\in L_{v_T,\mathrm{no}}$, then $\tilde v\le v_T\le\bar t-m^{-2}$, so the decision is correct.

The class statements follow from (ii)--(iii), Theorem~\ref{thm:lower}, and Theorems~\ref{thm:cl-universal} and~\ref{thm:qm-universal}.
\end{proof}

\begin{remark}[Complexity of the separating problem]\label{rem:where}
{The separating problem $L_{v_T}$ lies outside $\PSPACE$ and is decidable in quasi-polynomial space, unlike the undecidable witnesses of Theorem~\ref{thm:triangle}(iii).} A polynomial-space decider for $L_{v_T}$ would give an exponential-space decider for the unary language of $T$, which the space hierarchy theorem forbids. Nothing in the argument depends on the details of $T$ beyond Fact~\ref{fact:hierarchy}.
\end{remark}

\subsection{The role of the presentation}\label{sec:presentation}

The theories $\GG_v$ and $\GG^{\mathrm Q}_v$ are isomorphic to classical theory with $\Delta_3$ adjoined and to $\QC$ with $\Delta_3$ adjoined. With the standard gates on $\Delta_3$ (preparation of $\delta_1$, readout of the label and discard), the classes of these two theories are $\IP=\PSPACE$ (all of their messages are classical registers, so Theorem~\ref{thm:cl-universal} applies) and $\QIP=\PSPACE$ (Lemma~\ref{lem:QC}). The isomorphism $L_v$ changes the coordinates in which the verifier's gates are written: the measurement $(f_y,u-f_y)$, with integer coordinates in $\GG_v$, becomes the measurement with effects $(0,0,v)$ and $(1,1,1-v)$ on $\Delta_3$, which reads the label and tosses a coin of bias $v$ when the label is $2$. When $v$ is efficiently approximable, this instrument can be implemented efficiently with fair coins, and the separation disappears. We say that $v$ is \emph{efficiently approximable} if a rational number within $2^{-k}$ of $v$ can be computed in time $\poly(k)$.

\begin{proposition}[Efficient presentations give $\PSPACE$]\label{prop:presentation}
If $v\in(0,1)$ is efficiently approximable, then $\GIP(\GG_v)=\GIP(\GG^{\mathrm Q}_v)=\PSPACE$. More precisely, every $k$-message verifier for $\GG_v$, respectively for $\GG^{\mathrm Q}_v$, can be replaced by a $k$-message classical verifier with classical messages, respectively by a $k$-message verifier for $\QC$, whose value on every input is within $1/12$ of the original value.
\end{proposition}
\begin{proof}
The lower bound is Theorem~\ref{thm:lower}. For the upper bound, let $V$ be a verifier for $\GG_v$ with at most $L=\poly(n)$ gates. Transport the whole interaction through the isomorphism of Example~\ref{ex:Gv}: every system of $\GG_v$ becomes a classical system, every prover strategy becomes a classical randomised prover with memory (positive unit-preserving maps between classical cones are stochastic matrices), and conversely every classical randomised prover is the transport of a strategy of $\GG_v$. The verifier's gates become classical gates, the preparation of $\delta_1\in\Delta_3$, the discard, and the measurement with effects $(0,0,v)$ and $(1,1,1-v)$, and its messages become classical registers. Since the isomorphism preserves all probabilities, the value of $V$ over $\GG_v$-provers equals the value of the transported verifier $V_{\mathrm{can}}$ over classical provers. $V_{\mathrm{can}}$ is not yet a verifier for classical theory with $\Delta_3$ adjoined, because the measurement with effects $(0,0,v)$ and $(1,1,1-v)$ is not in the standard gate set. Fix $m=\lceil\log_2(12L)\rceil$, compute in time $\poly(m)$ a dyadic rational $\tilde v$ with $m$ binary digits and $|\tilde v-v|\le2^{-m}$, and let $V'$ be $V_{\mathrm{can}}$ with every such measurement replaced by the following circuit: read the label; if it is $2$, output the outcome of a coin of bias $\tilde v$, implemented by comparing $m$ fair coin tosses with the binary digits of $\tilde v$; otherwise output $0$. Then $V'$ is a polynomial-time uniform verifier for the theory $\GG'$ of classical theory with $\Delta_3$ adjoined, all of its messages are classical registers, and it has the same number of messages as $V$. The two measurements differ, on every input state of $\Delta_3$, by at most $|v-\tilde v|\le2^{-m}$ in statistical distance, and replacing one stochastic map in a classical probabilistic process by another changes the final distribution by at most the largest statistical distance between their outputs on a common input. Hence for every classical prover $P$, $|\Pr[V'_x\leftrightarrow P]-\Pr[V_{\mathrm{can},x}\leftrightarrow P]|\le L\,2^{-m}\le1/12$, and $\val_{\mathrm{classical}}(V'_x)$ is within $1/12$ of $\val_{\GG_v}(V_x)$. If $V$ witnesses $L\in\GIP(\GG_v)$, then $V'$ witnesses $L\in\GIPcl(\GG')$ with completeness $\ge7/12$ and soundness $\le5/12$, so $L\in\PSPACE$ by Proposition~\ref{prop:error} and Theorem~\ref{thm:cl-universal}.

For $\GG^{\mathrm Q}_v$ the argument is the same with $\QC$ in place of classical theory: the transported prover strategies are exactly the $\QC$-prover strategies, the transported verifier $V_{\mathrm{can}}$ is a $\QC$-verifier except for the coin of bias $v$, and $V'$ is a $\QC$-verifier. Replacing one measurement in the interaction by another whose flagged channel is at diamond-norm distance at most $\delta$ changes the final joint state by at most $\delta$ in trace norm, hence the acceptance probability against any fixed prover by at most $\delta/2$. The flagged channels of the two measurements in question differ on every input only in the weights, $v$ versus $\tilde v$, assigned to the two outcomes when the label is $2$, so their diamond-norm distance is at most $2|v-\tilde v|\le2\cdot2^{-m}$. Hence $\val_{\QC}(V'_x)$ is within $L\,2^{-m}\le1/12$ of $\val_{\GG^{\mathrm Q}_v}(V_x)$, and $L\in\GIP(\QC)=\QIP=\PSPACE$ by Lemma~\ref{lem:QC}.
\end{proof}

\begin{remark}[{Dependence on the verifier's gate set}]\label{rem:presentation}
Theorems~\ref{thm:triangle}, \ref{thm:Qtriangle} and~\ref{thm:computable} and Proposition~\ref{prop:presentation} concern one family of theories with a fixed operational content. Whether its interactive-proof class is $\PSPACE$, a class that strictly contains $\PSPACE$ and is witnessed by a problem decidable in quasi-polynomial space, or a class containing undecidable problems depends only on the real number $v$ that fixes the coordinates of the state space relative to those of the verifier's gates. The standing assumptions (T1)--(T4) cannot distinguish these cases, because they constrain the gates and not the state spaces. Any upper bound on $\GIP(\GG)$ for arbitrary messages therefore needs a hypothesis on the message state spaces in the coordinates in which the gates are specified, for instance that the maximum of an effect with efficiently computable coordinates over the states of a message system can be approximated efficiently. This is the sense in which the class depends on how the theory specifies its efficient operations. The construction is related to the observation of Adleman, DeMarrais and Huang that $\BQP$ with arbitrary real gate amplitudes contains undecidable problems \cite{ADH97}; there the noncomputable number sits in a gate, here it sits in the state space while the gates remain effective, which is why the phenomenon appears only with nonclassical messages and disappears for classical messages (Theorem~\ref{thm:cl-universal}) and for quantum verifiers (Theorem~\ref{thm:qm-universal}).
\end{remark}

\section{Discussion}\label{sec:discussion}

Table~\ref{tab:summary} collects the results. Throughout, $\GG$ is a theory satisfying (T1)--(T4).

\begin{table}[t]
\centering
\small
\renewcommand{\arraystretch}{1.12}
\begin{tabular}{@{}>{\raggedright\arraybackslash}p{0.34\textwidth}>{\raggedright\arraybackslash}p{0.44\textwidth}>{\raggedright\arraybackslash}p{0.16\textwidth}@{}}
\toprule
Setting & Result & Where\\
\midrule
Classical messages, any $\GG$ & $\GIPcl(\GG)=\PSPACE$, also for the public-coin subclass & Thm.~\ref{thm:cl-universal}\\
\addlinespace
Quantum verifier and messages, $\GG\supseteq$ quantum theory & $\val_\GG=\val_{\mathrm{quantum}}$; $\GIPqm(\GG;k)=\QIP(k)$; $\GIPqm(\GG)=\GIPqm(\GG;3)=\PSPACE$ & Thm.~\ref{thm:qm-universal}\\
\addlinespace
Instances & {specialises to the known equalities} $\IP=\PSPACE$ and $\QIP=\QIP(3)=\PSPACE$; $\GIP(\QC)=\QIP$ & Cor.~\ref{cor:instances}, Lem.~\ref{lem:QC}\\
\midrule
$\GG_v$, $v$ noncomputable & $\GMA(\GG_v)$ contains an undecidable problem; $\GIPcl(\GG_v)=\PSPACE\subsetneq\GIP(\GG_v)$ & Thm.~\ref{thm:triangle}\\
\addlinespace
$\GG^{\mathrm Q}_v\supseteq$ quantum and classical theory & $\GIPcl=\GIPqm=\PSPACE$; $\GIP\supsetneq\PSPACE$ for noncomputable $v$ & Thm.~\ref{thm:Qtriangle}\\
\addlinespace
$\GG_{v_T}$, $\GG^{\mathrm Q}_{v_T}$, $v_T$ computable & $\GIP\supsetneq\PSPACE$ unconditionally; the witness is in $\DSPACE(2^{O(\log^2n)})$ & Thm.~\ref{thm:computable}\\
\addlinespace
$\GG_v$, $\GG^{\mathrm Q}_v$, $v$ efficiently approximable & $\GIP=\PSPACE$ & Prop.~\ref{prop:presentation}\\
\bottomrule
\end{tabular}
\caption{Summary of results.}\label{tab:summary}
\end{table}

\paragraph{Message systems.}
A prover may compute with arbitrary systems of the theory. As long as the messages are classical registers, the class is $\PSPACE$, and as long as the verifier is a quantum circuit exchanging qubit registers, a prover with additional resources is no stronger than a quantum prover. We prove both statements by reducing the prover's behaviour on the message registers to a strategy of the corresponding subtheory: causality and positivity give classical conditional distributions for classical messages, and complete positivity together with the realisation theorem for quantum combs gives a quantum strategy for quantum messages. {The known bounds $\IP=\PSPACE$ and $\QIP=\PSPACE$ are inputs to these proofs.} Example~\ref{ex:unot} shows that the hypothesis on the composites of qubits in Theorem~\ref{thm:qm-universal} cannot be dropped.

\paragraph{Specification of efficient operations.}
The assumptions of the computation-in-general-theories literature constrain the gates and leave the state spaces free. Section~\ref{sec:beyond} shows that this allows a prover to communicate a noncomputable, or merely slowly computable, real number through the geometry of a message, even in theories isomorphic to quantum theory with classical systems, {whereas the isomorphic theory equipped with its standard gate set has $\GIP=\PSPACE$.} A polynomial-space upper bound for arbitrary messages therefore needs a hypothesis relating the message state spaces to the coordinates of the gates; Proposition~\ref{prop:presentation} gives one such hypothesis for the theories at hand. For theories such as boxworld, whose message state spaces are polytopes with rational vertices in the given coordinates, no such obstruction is visible, but we do not know whether $\GIP(\text{boxworld})=\PSPACE$.

\paragraph{Open questions.}

\begin{question}[Effectiveness hypotheses]\label{q:effective}
Which hypothesis on the message state spaces, formulated in the coordinates of (T4), implies $\GIP(\GG)=\PSPACE$ for a natural class of theories including theories that are not isomorphic to quantum or classical theory? Is $\GIP(\text{boxworld})=\PSPACE$? No analogue of the realisation theorem for quantum combs is known for boxworld, so the strategies of a boxworld prover have no known finite description of the kind that underlies the quantum upper bound \cite{GW07,JJUW11}.
\end{question}

\begin{question}[Constant numbers of classical messages]\label{q:constant}
With a classical verifier, a constant number $k\ge2$ of classical messages gives exactly $\AM$ (Fact~\ref{fact:AM}). With a nonclassical verifier and classical messages, is $\GIPcl(\GG;k)=\AM$ for every theory $\GG$? {Already for $\GG=\QC$, a positive answer would imply $\BQP\subseteq\AM$, since the verifier may ignore the prover.} The verifier's transition weights are path sums rather than uniform counts, and the set-lower-bound protocol of Goldwasser and Sipser \cite{GS86} does not obviously apply to them.
\end{question}

\begin{question}[Several provers]\label{q:multi}
With two or more separated provers that share a state of a composite system, classical messages no longer make the class theory-independent: classical provers with shared randomness give $\MIP=\NEXP$ \cite{BFL91}, entangled quantum provers give $\MIPs=\RE$ \cite{JNV21}, and provers restricted only by the no-signalling condition give exactly $\PSPACE$ for two provers and one round (the inclusion of $\PSPACE$ is from Ref.~\cite{IKM09} and the converse from Ref.~\cite{Ito10}) and $\EXP$ for polynomially many provers \cite{KRR14}. The composition rule of the theory then influences the class. What are the multi-prover classes of boxworld and of other theories satisfying (T1)--(T4)?
\end{question}

\begin{question}[Theories without local tomography]\label{q:LT}
Real quantum theory violates (T1) (Example~\ref{ex:realQ}). {Using real simulations \cite{MMG09}, McKague proves $\QIP_{\R}(k)=\QIP(k)$ for every $k$ \cite[Theorem~1]{McK13}, so its interactive-proof class is $\PSPACE$.} Which parts of Theorems~\ref{thm:cl-universal} and~\ref{thm:qm-universal} survive when local tomography is replaced by a weaker structural assumption on composites?
\end{question}

\begin{question}[Round reduction]\label{q:rounds}
Quantum theory compresses every interactive proof to three messages \cite{KW00}, whereas classical theory does not unless $\PSPACE=\AM$. Since classical theory satisfies (T1)--(T4), no round-reduction theorem can follow from the standing assumptions alone. Which additional hypotheses on a theory, for instance the availability of coherent operations on a qubit to the verifier, imply $\GIP(\GG)=\GIP(\GG;3)$?
\end{question}

\section*{Acknowledgements}
The author acknowledges support from the Agency for Science, Technology and Research (A*STAR), Singapore, and from the Hartree Postdoctoral Fellowship at the Joint Center for Quantum Information and Computer Science (QuICS), University of Maryland. The author thanks Antonios Varvitsiotis and Rahul Jain for preliminary discussions during the early stages of this project.

\section*{Project history and use of generative AI}
This project began in 2023, following the work of Canyakmaz, Lin, Piliouras and Varvitsiotis on multiplicative updates over symmetric cones \cite{CLPV23}, and included preliminary discussions with Antonios Varvitsiotis and Rahul Jain. The author has continued developing the project since then. AI tools were used more recently to help polish the draft and, during revision, to assist with reference checking and critical examination of the exposition and mathematical arguments.

\bibliographystyle{alpha}
\bibliography{refs}

\appendix
\section{The unary space hierarchy theorem}\label{app:hierarchy}

We recall the proof of Fact~\ref{fact:hierarchy}, which is the standard space hierarchy diagonalisation \cite{HLS65,AB09} carried out on unary inputs.

Let $g(n)=2^{n^2}$, which is space constructible: a machine can mark off $g(n)$ cells on input $1^n$ using space $O(g(n))$. Fix an enumeration $M_1,M_2,\dots$ of deterministic Turing machines with a read-only input tape and one work tape over fixed input and work alphabets, in which every machine occurs infinitely often (for instance, enumerate pairs consisting of a machine description and a padding integer) and in which the description of $M_n$ has $O(n)$ bits and is computable from $1^n$ in space $O(n)$. Define
\[
T=\{n\in\N: M_n\text{ on input }1^n\text{, with work space limited to }g(n)\text{ cells, halts and rejects}\}.
\]
A universal machine decides $T$ on input $1^n$ as follows: mark $g(n)$ cells, compute the description of $M_n$, and simulate $M_n$ on $1^n$ within the marked region, using a counter of $O(g(n)+n)$ bits to abort the simulation if it runs for more than $(n+2)\,|\Gamma|^{g(n)}\,g(n)\,|Q|$ steps, where $\Gamma$ and $Q$ are the work alphabet and the state set of $M_n$ (a machine that runs longer within bounded space repeats a configuration and never halts); accept if and only if the simulation halts and rejects. Because the work alphabet is fixed, the simulated tape is stored in $O(g(n))$ cells; the machine description, state and head positions require only $O(n+\log g(n))$ additional cells. We therefore obtain a uniform $O(g(n))$ space bound, with a constant independent of $M_n$, and by tape compression $\{1^n:n\in T\}\in\DSPACE(g(n))$.

Suppose that $T$, as a unary language, were decided by a machine $M$ using space $f(n)\le2^{cn}$ for some constant $c$; converting $M$ to the fixed work alphabet costs at most a constant factor in space, so we may assume that $M$ is one of the enumerated machines. Then $M=M_n$ for infinitely many $n$. Since $2^{cn}\le2^{n^2}$ for all large $n$, for large $n$ with $M_n=M$ the simulation of $M$ on $1^n$ within $g(n)$ cells is faithful and halts, because $M$ halts on every input. Then $n\in T$ if and only if $M$ rejects $1^n$, if and only if $1^n\notin L(M)=\{1^n:n\in T\}$, a contradiction. Hence $\{1^n:n\in T\}\in\DSPACE(2^{n^2})\setminus\bigcup_c\DSPACE(2^{cn})$. Changing finitely many elements of $T$, as in Theorem~\ref{thm:computable}, affects neither statement.

\end{document}